\documentclass[onecolumn,prd,longbibliography,aps,superscriptaddress]{revtex4-1}

\usepackage{amsfonts,amsthm,amsmath,amssymb,bbold, upgreek, dsfont, tipa, epsfig, color, inconsolata, framed, mathtools, sidecap, graphics, graphicx, wrapfig, overpic, hyperref, mathrsfs, textcomp, verbatim, subfigure, hyperref, multirow}

\usepackage[english]{babel}
\usepackage{rotating} 

\usepackage[utf8]{inputenc}
\usepackage[T1]{fontenc}
\usepackage{hhline}
\usepackage[ruled,vlined]{algorithm2e}

\usepackage{tikz, pgfplots} 

\usepackage{tikz}

\date{\today}

\newcommand{\vect}[1]{\boldsymbol{#1}}

\newcommand{\eps}{\epsilon}

\newcommand{\de}{\delta}

\newcommand{\CellDom}{\mathscr{D}}

\DeclareMathOperator*{\argmax}{arg\,max}

\newcommand{\p}{\partial}
\newcommand{\dd}{\text{d}}

\newcommand{\xv}{\vect{x}}

\newcommand{\cv}{c}
\newcommand{\bv}{b}
\newcommand{\fv}{f}

\newcommand{\cvec}{\vec{c}}
\newcommand{\bvec}{\vec{b}}
\newcommand{\fvec}{\vec{f}}

\newcommand{\cn}{A}
\newcommand{\cnvec}{\vec{\cn}}

\newcommand{\Om}{\Omega}

\newcommand{\DelX}{\nabla^2}
\newcommand{\NabX}{\nabla}

\newcommand{\DAPI}{\emph{DAPI} }
\newcommand{\HER}{\emph{HER2} }
\newcommand{\ER}{\emph{ER} }

\newtheorem{prop}{Proposition}[section] 
 
\newtheorem{remark}{Remark}[section]

\begin{document}

\title{Simulated Ablation for Detection of Cells Impacting Paracrine Signalling in Histology Analysis.}

\author{Jake P. Taylor-King}
\affiliation{Mathematical Institute, University of Oxford, Oxford, OX2 6GG, UK}
\affiliation{Department of Integrated Mathematical Oncology, H. Lee Moffitt Cancer Center and Research Institute, Tampa, FL, USA}

\author{Etienne Baratchart}
\affiliation{Department of Integrated Mathematical Oncology, H. Lee Moffitt Cancer Center and Research Institute, Tampa, FL, USA}

\author{Andrew Dhawan}
\affiliation{Department of Oncology, University of Oxford, Old Road Campus Research Building, Roosevelt Drive, Oxford, OX3 7DQ, UK}

\author{Elizabeth A. Coker}
\affiliation{Cancer Research UK Cancer Therapeutics Unit, The Institute of Cancer Research, London, UK}

\author{Inga Hansine Rye}
\affiliation{Department of Genetics, Institute for Cancer Research, Oslo University Hospital Radiumhospitalet, Oslo 0424, Norway}

\author{Hege Russnes}
\affiliation{Department of Genetics, Institute for Cancer Research, Oslo University Hospital Radiumhospitalet, Oslo 0424, Norway}
\affiliation{Department of Pathology, Oslo University Hospital, Oslo 0424, Norway}

\author{S. Jon Chapman}
\affiliation{Mathematical Institute, University of Oxford, Oxford, OX2 6GG, UK}

\author{David Basanta}
\affiliation{Department of Integrated Mathematical Oncology, H. Lee Moffitt Cancer Center and Research Institute, Tampa, FL, USA}

\author{Andriy Marusyk}
\affiliation{Department of Cancer Imaging and Metabolism, H. Lee Moffitt Cancer Center and Research Institute, Tampa, FL, USA}

\date{\today}

\begin{abstract}
Chapman, Basanta and Marusyk are joint senior author. 
\newline
Contact email: jake.taylor-king@sjc.ox.ac.uk.
\newline\newline
Intra-tumour phenotypic heterogeneity limits accuracy of clinical diagnostics and hampers the efficiency of anti-cancer therapies. Dealing with this cellular heterogeneity requires adequate understanding of its sources, which is extremely difficult, as phenotypes of tumour cells integrate hardwired (epi)mutational differences with the dynamic responses to microenvironmental cues. The later come in form of both direct physical interactions, as well as inputs from gradients of secreted signalling molecules. Furthermore, tumour cells can not only receive microenvironmental cues, but also produce them. Despite high biological and clinical importance of understanding spatial aspects of paracrine signaling, adequate research tools are largely lacking. Here, a partial differential equation (PDE) based mathematical model is developed that mimics the process of cell ablation. This model suggests how each cell might contribute to the microenvironment by either absorbing or secreting diffusible factors, and quantifies the extent to which observed intensities can be explained via diffusion mediated signalling. The model allows for the separation of phenotypic responses to signalling gradients within tumour microenvironments from the combined influence of responses mediated by direct physical contact and hardwired (epi)genetic differences. The differential equation is solved around cell membrane outlines using a finite element method (FEM). The method is applied to a multi-channel immunofluorescence \emph{in situ} hybridization (iFISH) stained breast cancer histological specimen and correlations are investigated between: \HER gene amplification; \HER protein expression; and cell interaction with the diffusible microenvironment. This approach allows partial deconvolution of the complex inputs that shape phenotypic heterogeneity of tumour cells, and identifies cells that significantly impact gradients of signalling molecules.

\end{abstract}

\maketitle 

\section{Introduction}


Phenotypic heterogeneity of malignant cells within tumours represents a major clinical challenge as it complicates diagnosis and underpins therapy resistance. This heterogeneity arises as the result of interplay between: a) cell-intrinsic differences stemming from genetic heterogeneity and stable epigenetically defined states; b) stochastically arising variability in gene expression; c) environmental inputs in form of physical forces, contact-mediated signals from neighbouring cells, the extracellular matrix (ECM), and diffusible signals in form of gradients of growth factors, cytokines, oxygen and metabolites \cite{Marusyk_2012}. Understanding exact sources of variability of clinically relevant phenotypic features is of paramount importance. However, deconvolution of the relative impact of inputs that shape phenotypes is extremely challenging as we lack appropriate research tools.



When considering anti-cancer therapeutics, the primary issue that has recently emerged is that heterogeneity is a critical biomarker of tumour prognosis, as greater heterogeneity provides a clear advantage in the face of the evolutionary bottleneck of anti-cancer therapy --- there are simply more paths available to facilitate the generation of resistance \cite{Maley_2006}. The design of therapies must consider the existing heterogeneity in a tumour and account for the most aggressive cells that may exist; and in failing to do so, one may select for pre-existing resistant cells \cite{Martinez_2016}. Furthermore, a critical generator of phenotypic heterogeneity involves microenvironmental variability within the tumour, as a varying environment contributes greatly to different modes of adaption, upon spatially organised subgroups of cells \cite{Swanton_2012}. In fact, it has been shown that heterogeneity of oxygen distribution within a tumour leads to the adaption of subsets of cells in microenvironmental niches to the hypoxic microenvironment, which has been shown to result in poorer prognosis and more aggressive tumours \cite{Lunt_2009}. Lastly, heterogeneity has also been implicated as an important consideration in the field of immunotherapy, where neo-antigen generation is a function of polypeptide heterogeneity, which is a critical determinant of the success of immunotherapy, and reflects underlying genetic heterogeneity \cite{Schumacher_2015}.

Microscopy imaging of histological specimens is widely and routinely used in clinical diagnostics of cancers, as well in experimental studies aiming to understand the underlying biology and responses to anti-cancer therapies. After fixation and placement on glass slides, tissues are subjected to chromogenic or fluorescent staining using chemical or antibody-based stains\footnote{Note that once fixation has taken place, using these cells for sequencing is no longer an option.}. Staining intensity reflects concentration of the chemical moiety, to which the stain binds. Therefore, digitisation of chromogenic or fluorescent signals allows quantification of concentration of the these chemicals. Since histological slides retain spatial information, they are suitable for the analysis of not only cellular phenotypes and genotypes, but also microenvironmental factors that shape cellular heterogeneity \cite{Heindl_2015}. Statistical approaches have been focused on the task on feature extraction, feature selection, and dimension reduction of large histological images. This field is sometimes known as whole-slide imaging (WSI). Feature extraction can happen at both at, a pixel-level, largely uninterpretable by a human (e.g., pixel intensities compared to neighbours), and at an object-level, recovering features that a pathologist would naturally be interested in (e.g., circularity of cell nuclei, location of blood vessels). Feature selection (along with the preceding staining procedure) can be carried out depending on the histopathology of the specific disorder, or based on dimensionality reduction (e.g., using principle component analysis) and irrelevant features can be ignored. Feature analysis is a statistical/machine learning problem, where the aim is to link the collective cell (or pixel) properties to macroscopic disorders/clinical outcomes \cite{Ghaznavi_2013}. Modern trends include considering ecologically motivated spatial statistics \cite{Heindl_2015, Natrajan_2016}. Histology analysis is both cheap and also highly clinically relevant, and a pathologist can diagnose based off an image. However, a histology slide is static and one is seldom able to elucidate any mechanism underlying observations.



Using mathematical modelling, we present a quantitative approach to calculate how each cell alters the local microenvironment using only histological images. Our approach allows us specify two things. First, to what degree a cell is either absorbing from, or secreting into, a local diffusible signalling environment; and second, how much of the observed staining intensities are explainable via diffusion. Therefore, when one stains a histology slide to measure the expression of a specific protein, one is then able to quantify the extent to which that protein contributes to microenvironmental signalling. The method presented is based on postulating that expression levels (as determined by staining intensity) of targets of analyses are determined by an effect of a field of secreted environmental signalling molecules. This signalling field (SF) abstraction integrates all of the secreted factors that impact the expression of a particular phenotypic trait in a paracrine manner (cytokines, gases, metabolites). Our approach ignores physical interactions between cells (e.g., force interactions through the ECM) and relies on availability of multi-channel staining data. 




Substantial mathematical modelling efforts have been focused on understanding the impact of the microenvironment on phenotypic heterogeneity using biophysical principles; common themes include diffusion of cytokines, and construction of chemical reaction networks. Efforts at modelling the microenvironment in a spatial setting have largely been focused on forward modelling, where one makes assumptions and rules for a model, initiates the model in some starting configuration, and then evolves it forward in time. Due to the complexity of the models involved, cellular automata approaches have often been utilised \cite{Anderson_2009, Anderson_2006, Pickup_2013, Basanta_2013, Anderson_2009b}. Such approaches, depending on the complexity of the rules built into the model, are able to reproduce many experimental/clinical observations, and in some cases are capable of making experimentally validated predictions \cite{Araujo_2014}. Unfortunately, mathematical modelling approaches are limited by: the frequent need for model iteration, where model components are added and removed; the complexity of rules needed to describe the behaviour of biological systems, such that underlying assumptions may often be untestable; and the experimental difficulty of obtaining relevant measurements required for adequate parametrisation of the underlying mathematical model. Our method should aid with future modelling as then one can ``rule in'' critical determinants.


Our approach is centred on mimicking the process of \emph{ablation} within the fields of neuroscience \cite{Xu_2016, Bono_2005, Wang_2010} and embryonic development \cite{Bargmann_1995}; this is where one kills or disables a single neuron to observe how the remaining system behaves. Each cell is considered a \emph{functional unit} that changes the SF from an implied \emph{baseline} value, to the \emph{observed} value of the staining intensity. The \emph{baseline} value we calculate as the SF at the cell's location \emph{were the cell not present}; this calculation is performed by solving a steady state diffusion equation (Poisson's equation) with decay. Since parameter estimation is likely impossible, we calculate the signal staining intensities that we expect based on the postulate of all of the variability coming from the impact of the signalling field.  We then compare the expected (\emph{baseline}) staining intensity of each cell with the experimentally measured (\emph{observed}) values. Differences between the \emph{baseline} and \emph{observed} values are interpreted as \emph{cell impact}: how a cell alters the SF. We test applicability of our approach using histological samples of breast cancer.

The paper is structured as follows: in Section \ref{sec_method_overview}, we present the mathematical details behind our method including: assumptions made and parameter selection. Section \ref{sec_molpath_results} discusses our method applied to a multi-channel immunofluorescence \emph{in situ} hybridization (iFISH) stained breast cancer data set, in which \HER gene and \HER protein expression were studied concurrently. Section \ref{sec_toy_system} lays out a toy problem to demonstrate our approach when one uses artificially generated data. In Section \ref{sec_discussion}, we conclude and discuss potential future work.


\section{Method Applied To Paracrine Signalling}\label{sec_method_overview}

Our approach consists of two stages: a mathematical modelling step (Section \ref{subsec_math_model}), and a parameter selection step (Section \ref{subsec_param_sel}). We first pose a class of diffusion models governed by partial differential equations (PDEs) for the description of SF in the extracellular space. This class of models has a number of free parameters: the effective diffusion constant; the rate of decay and the type of boundary condition posed on the cell surfaces. We then consider model selection by asking: what is the discrepancy between the expression of target being analysed (the \emph{observed} staining intensity) and the expected SF \emph{were the cell absent from the histological slide} (the \emph{baseline} staining intensity). Using the assumption that the SF should account for most of the variance in signal that the cell produces, we carry out model selection over the space of possible model parameterisations. Finally, we then have a \emph{baseline} staining signal intensity and an \emph{observed} signal intensity for each cell.


\subsection{Mathematical Modelling}\label{subsec_math_model}

\begin{figure}[t]
  \centering
  \vspace{1em}
  \begin{overpic}[width=0.49\textwidth]{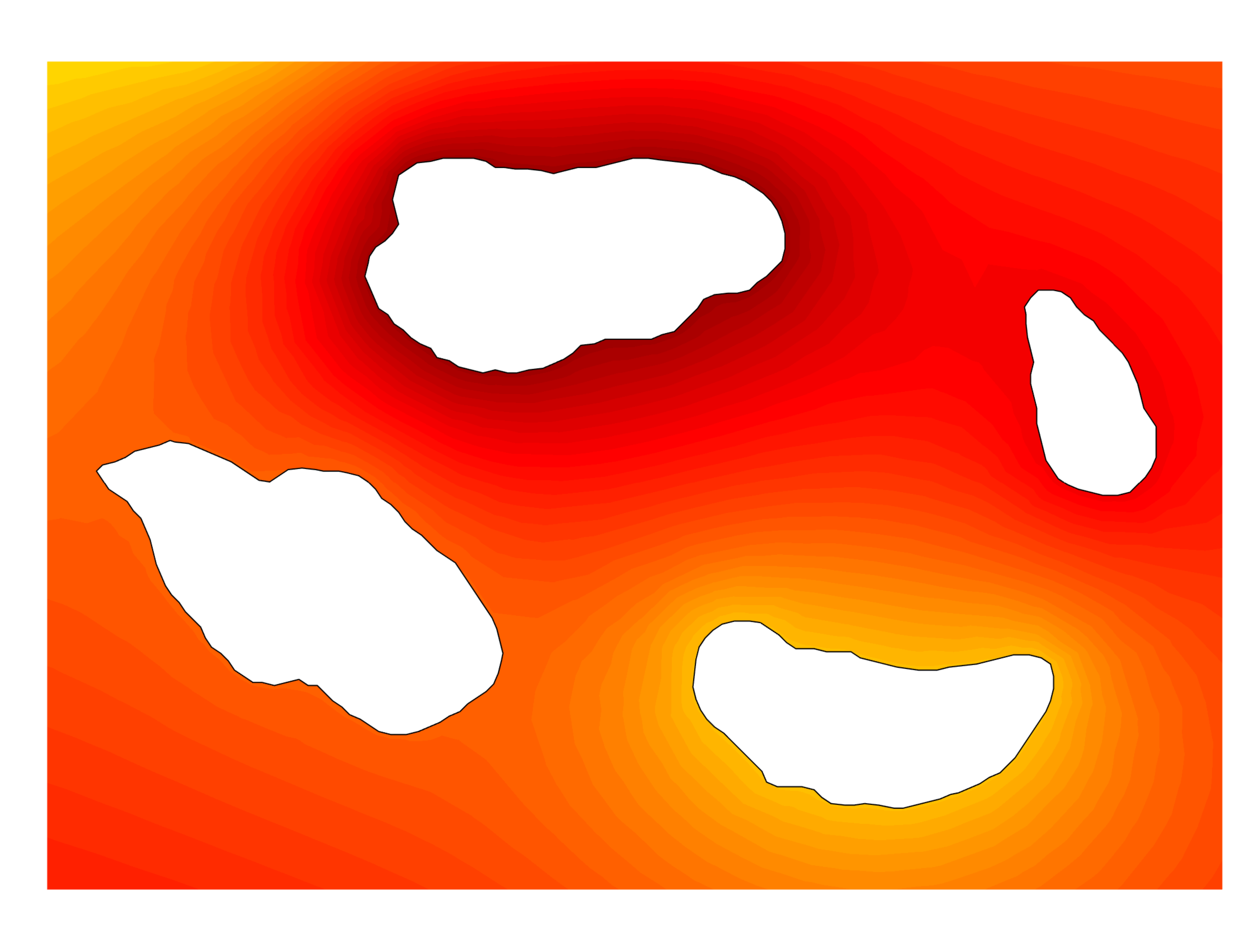}
  	\put(5,66){\Large a.)}
	\put(42,57){\small Cell $j$}
	\put(41,53){\small $(\cv_j, \CellDom_j)$}
	\put(44.2,47.75){\Large \rotatebox{-60}{$\rightarrow$}}
	\put(41.5,43){\small $\vect{n}_j$}
	\put(20,31){\small Cell $k$}
	\put(19,27){\small $(\cv_k, \CellDom_k)$}
	\put(30,36){\Large \rotatebox{50}{$\rightarrow$}}
	\put(33.5,37){\small $\vect{n}_k$}
	\put(67,19){\small Cell $i$}
	\put(66,15){\small $(\cv_i, \CellDom_i)$}
	\put(50.9,23.5){\Large \rotatebox{-185}{$\rightarrow$}}
	\put(50,18){\small $\vect{n}_i$}
	\put(8,7){\LARGE $\Om \subset \mathds{R}^2$}
  \end{overpic} 
   \begin{overpic}[width=0.49\textwidth]{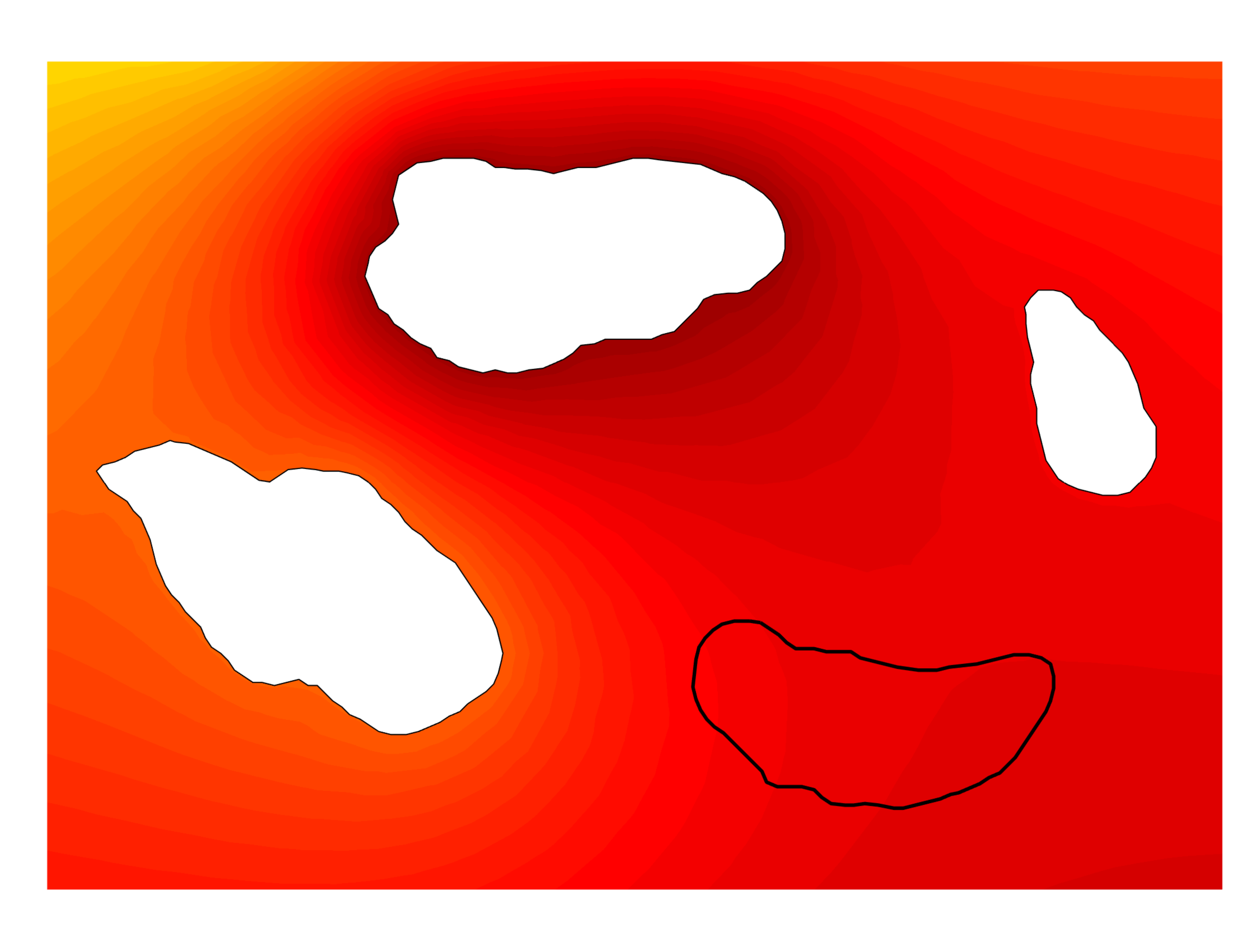}
   	\put(5,66){\Large b.)}
	\put(42,57){\small Cell $j$}
	\put(41,53){\small $(\cv_j, \CellDom_j)$}
	\put(44.2,47.75){\Large \rotatebox{-60}{$\rightarrow$}}
	\put(41.5,43){\small $\vect{n}_j$}
	\put(20,31){\small Cell $k$}
	\put(19,27){\small $(\cv_k, \CellDom_k)$}
	\put(30,36){\Large \rotatebox{50}{$\rightarrow$}}
	\put(33.5,37){\small $\vect{n}_k$}
	\put(8,7){\LARGE $\Om_i = \Om \cup \CellDom_i \subset \mathds{R}^2$}
  \end{overpic} 
  \caption{\footnotesize{Mathematical idealisation of cells on a pathology slice. (a.) Original domain $\Om$, and (b.) modified domain $\Om_i = \Om \cup \CellDom_i$.}}
  \label{fig_cell_domain}
\end{figure}
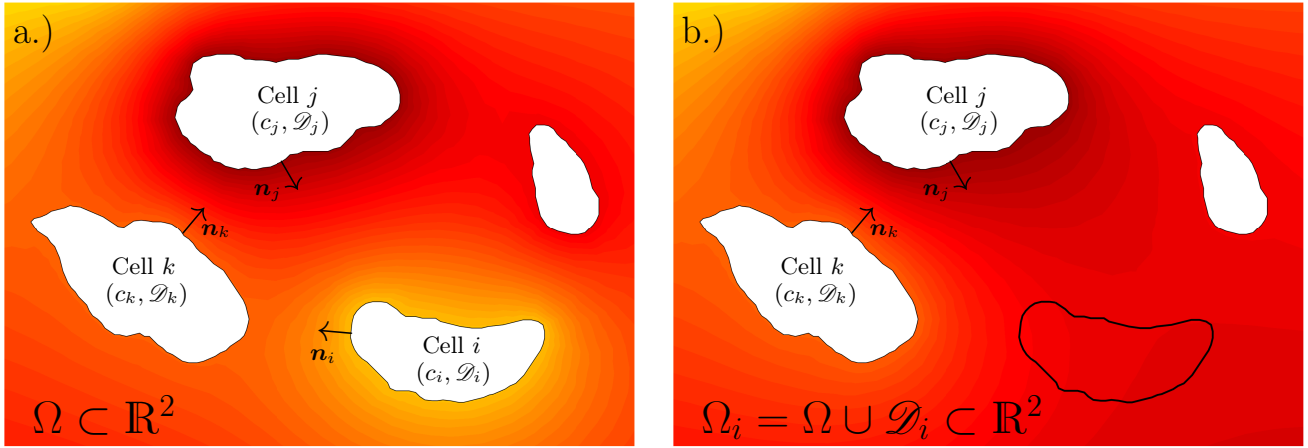

To model paracrine signalling, we assume that cells communicate via a diffusible species present in the extracellular domain called the Signalling Field (SF), with concentration profile $u = u(\xv)$, which degrades at rate $\lambda>0$. Cells are stained with a probe that binds internally, the chromogenic or fluorescent intensity is not necessarily representative of the SF; however we specify that production of the SF is linearly related to the stain intensity. We do not model other signalling processes such as direct signalling or forces exerted between cells via direct contact (such as adherent or gap junctions) or indirect contact via the ECM.

A histology slide is a 2-dimensional slice though a 3-dimensional tissue. In reality, the cells visible on this slide would be interacting with cells above and below the side (before sectioning the tissue). However, for the purposes of our reconstruction we consider only cells visible in the slide, and consider the signalling field to diffuse in 2-dimensions only. In principle a 3-dimensional analysis could be performed by considering multiple adjacent pathology slides.

We consider $N$ non-overlapping cells labelled $i\in\mathcal{N} = \{1,\dots, N\}$ that occupy volumes $\CellDom_i \subset \mathds{R}^2$ \footnote{Thus $\CellDom_i \cap \CellDom_j = \emptyset$ for $i \neq j$.} (see Figure \ref{fig_cell_domain}(a)), so that extracellular domain is $\Om  = \mathds{R}^2\setminus \bigcup_{i=1}^N \CellDom_i$. We consider the case in which the timescale of diffusion is faster than any other timescale of interest so that the SF is in steady state. The concentration profile $u$ is then governed by 
\begin{equation}\label{eq_gen_chem_prof}
\DelX u  - \alpha^2 u = 0  \text{ in }\Om  \, ,
\end{equation}
where $\alpha^2 = \lambda / \kappa$ and $\kappa>0$ is the diffusion coefficient.

We suppose that the production/absorption of the SF is proportional to the difference in the average cellular stain intensity $\cv_i$ (averaged over the cell) and the concentration of the SF at the cell boundary. Thus we write \footnote{In principle, we could use the spatially varying stain intensity in equation (\ref{eq_gen_chem_prof_bc}) rather than averaging over the cell. Measuring and quantifying a heterogeneous protein concentration within a cell would be experimentally demanding and may require extra analysis, e.g., location of internal cell apparatus. This spatial averaging is often done automatically by WSI software.} 
\begin{equation}\label{eq_gen_chem_prof_bc}
  \vect{n}_i \cdot \NabX u  =  \gamma (\cv_i   - u) \text{ on }\xv\in\p\CellDom_i\, .
\end{equation}


The parameter $\gamma\in(0,\infty)$ is a measure of how strong the cellular response is to a difference between the local SF field and the ``target'' value $\cv_i$. When $\gamma = 0$, the cell does not react to the signalling field at all; as $\gamma\to\infty$ the cell actively absorbs/secretes the SF as fast as it needs to in order to maintain the interior of the cell at the fixed \emph{observed} staining intensity $\cv_i$.

The problem defined by equations \eqref{eq_gen_chem_prof}--\eqref{eq_gen_chem_prof_bc} is well-posed and would allow us to recreate the extracellular SF were parameters ($\alpha,\gamma$) known. However, it was our aim to learn how each cell modifies the SF. To this end, for each cell $i\in\mathcal{N}$, we will determine the difference between the \emph{observed} staining intensity ($\cv_i$) and the expected signal \emph{were the cell not present}. Therefore, we define a new SF problem with cell $i$ removed. Denoting the resulting concentration profile by $u_i$, we then have
\begin{equation}\label{eq_new_chem_prof} 
\DelX u_i  -  \alpha^2 u_i  = 0\text{ in }\Om_i \, ,
\end{equation}
and $\Om_i = \Om  \cup \CellDom_i = \mathds{R}^2 \setminus \bigcup_{j\neq i} \CellDom_j$. Thus the SF now diffuses in the region occupied by cell $i$ also; the boundary conditions on the remaining cells stay the same
\begin{equation}\label{eq_new_chem_prof_bc}
 \vect{n}_j \cdot \NabX u_i    = \gamma (\cv_j - u_i ) \text{ on }\xv\in\p\CellDom_j\, ,
\end{equation}
for all $j = 1,\dots,i-1,i+1,\dots,N$. The domain for the modified problem is illustrated in Figure \ref{fig_cell_domain}(b). The problem is also well posed (see Appendix \ref{app_exist_unique}). The \emph{baseline} SF were cell $i$ not present is then defined as the mean value of the SF over the cell
\begin{equation}\label{eq_baseline_def_1}
\bv_i := \frac{1}{ \vert \CellDom_i \vert }  \int_{\CellDom_i} u_i (\xv) \dd \xv \, .
\end{equation}
We could also average over the cell boundary, however this gives no difference in results qualitatively. We define the \emph{cell impact} to be 
\begin{equation}\label{eq_defn_cellAct}
\fv_i := \cv_i - \bv_i \, ,
\end{equation}
which can be interpreted as follows: if $\fv_i>0$, the cell is secreting factors locally into the SF; and if $\fv_i<0$, then the cell is absorbing factors from the SF.

We solve the PDE systems \eqref{eq_new_chem_prof}--\eqref{eq_new_chem_prof_bc} using a finite element method; see Appendix \ref{app_numerical_approach} for details.

\subsection{Parameter Selection}\label{subsec_param_sel}

Parameter selection is the biggest challenge to implementation of our method. For metabolic processes, diffusion constants and decay rates are frequently known, for example, oxygen, glucose, etc. However, in our case, the specific chemicals comprising the signalling field are unknown (we note again that this is not the chemical stained for, but a hypothesised downstream factor). 


Rather than trying to estimate $\vec{p} = (\alpha, \gamma)$ from experiments, we choose $\vec{p}$ to give a best fit to the data in the following sense. For each stain, we measure the \emph{observed} staining intensities for each cell $\cvec = \{\cv_1, \dots, \cv_N\}$. For a fixed set of model parameters $\vec{p}$, the method described above can be used to generate the \emph{baseline} intensities $\bvec = \{ \bv_1 ( \vec{p}),\dots, \bv_N ( \vec{p} )  \}$. 

We then choose parameters $\vec{p}_*$ such that the coefficient of determination (denoted $R^2$) is maximised, i.e.,
\begin{equation}
\vec{p}_* = \argmax_{ \vec{p} = (\alpha, \gamma)  }  R^2 \, ,
\end{equation}
where 
\begin{equation}\label{eq_Rsq_form}
R^2 = 1 - \frac{  \sum_{i=1}^N  (\cv_i - \bv_i )^2  }{  \sum_{i=1}^N  (\cv_i - \bar{c} )^2    } \, ,
\end{equation}
with
\begin{equation}
\bar{c} = \frac{1}{N}\sum_{i=1}^N \cv_i \, .
\end{equation}
The coefficient $R^2$ measures the fraction of the variance of the stain intensities that can be accounted for by our signalling field model.

\section{Experimental Data}\label{sec_molpath_results}

We apply our approach to the analysis of breast cancer tissue sections stained with nuclear stain \DAPI that reflects cellular DNA content, and two important proteins targets:\footnote{These were done with Immunohistochemistry (IHC) derived methods.}  \HER, a receptor tyrosine kinase that is amplified in subset of breast cancers and is considered to be a potent `driver' gene \cite{Moasser_2007}, and \emph{Oestrogen Receptor (ER)}, a nuclear receptor that mediates cellular response to oestrogen signalling. Both \HER and \ER staining are expected to have profound influence on cells phenotypes. Additionally, using fluorescence \emph{in situ} hybridization (FISH), data regarding the amplification status of the \HER gene is available.

The data used collected in compliance with the Declaration of Helsinki and was approved by the regional ethics committee (REK S-06495b) and the institutional review board of Oslo University Hospital Radiumhospitalet (IRB 2006-53). Full details of the experimental protocol can be found in Refs.~\cite{Trinh_2014,Almendro_2014}.

The cell outlines were found using the GoIFISH software \cite{Trinh_2014}, the pixels that formed the cell outlines were down sampled by a factor of $8$, so that the cell edges were given by polygons. For the FEM scheme, nodes within domain $\Om$ were placed using Halton node placing \cite{Fornberg_2015}. To simplify parameter searches in $\gamma\in (0,\infty)$ (which crosses many orders of magnitude) we mapped it to the bounded variable $\beta = \gamma/(\gamma + 1) \in (0,1)$.

From the method in Section \ref{sec_method_overview}, we maximise $R^2$ given in equation \eqref{eq_Rsq_form} for each of the 3 stains for 3 histology slides. The corresponding values of $\alpha_*$ and $\gamma_*$ are found in Table \ref{table_BCA_params}. From Table \ref{table_BCA_params}, we see that the $R_{\max}^2$ is small for \DAPI as expected: DNA content should be independent of diffusible signalling components. In contrast, for the \HER and \ER stains, Table \ref{table_BCA_params} suggests that the majority of the variance in the \emph{observable} data set can be explained by SF. 

Note that the large values of $\gamma_*$ indicate a strong coupling between the \HER staining intensity and the SF. It is reassuring that $R^2_{\max}$ and $\alpha_*$ are similar for the same stain across different slides. The range of values of $\gamma_*$ are less consistent. 

%

\begin{table}[h]
\centering
 \begin{tabular}{|| c | c || c | c | c || c | c | c || c | c | c ||} 
 \hline
 \multirow{2}{*}{$\quad${\bf Stain}$\quad$} & \multirow{2}{*}{$\quad${\bf Region}$\quad$} & \multicolumn{3}{ c ||}{{\bf Sample 1,} $N = 479$} & \multicolumn{3}{ c ||}{{\bf Sample 2,} $N = 381$} & \multicolumn{3}{ c ||}{{\bf Sample 3,} $N = 524$}  \\
 
 &  & $\alpha_*$ & $\gamma_*$ & $R_{\max}^2$ & $\alpha_*$ & $\gamma_*$ & $R_{\max}^2$ & $\alpha_*$ & $\gamma_*$ & $R_{\max}^2$  \\ [0.2ex] 
 \hline
DAPI & Nucleus        & $0.008$ & $0.15$ & $17\%$    & $0.006$ & $0.04$ & $12\%$   & $0.006$ & $0.05$ & $21\%$  \\  
HER2 & Membrane  &  $0.000$ & $\infty$& $61\%$    &  $0.023$ & $\infty$ & $60\%$ &  $0.010$ & $6.75$ & $67\%$ \\ 
ER & Nucleus           &  $0.012$ & $0.14$ & $51\%$   &  $0.012$ & $1.11$ & $45\%$ &  $0.003$ & $0.07$ & $50\%$ \\ 
 \hline
 \end{tabular}
 \caption{Maximum variance reduction possible for each stain and corresponding parameters for BCA data set.}
  \label{table_BCA_params}
\end{table}

In Figure \ref{fig_breast_spatial_staining_patterns}, we plot the \emph{observed}, \emph{baseline} and \emph{cell impact} spatial staining patterns for the \DAPI, \HER and \ER stains for sample 1. Additionally, we also plot the solution to equation \eqref{eq_gen_chem_prof} using the $\vec{p}_*$ parameter set (without removing any cells). This allows us to visualise the hypothesised SF. Plots relating to the same stain use the same scaled colour bar. The general trend one finds is that the \emph{observed} data set is very heterogeneous with regards to staining intensity, the \emph{baseline} data set looks like a smoothed version of the \emph{observed} data set, and the \emph{cell impact} data set shows variations around the \emph{baseline}.

\begin{figure}[t]
    \begin{overpic}[width=0.23\textwidth]{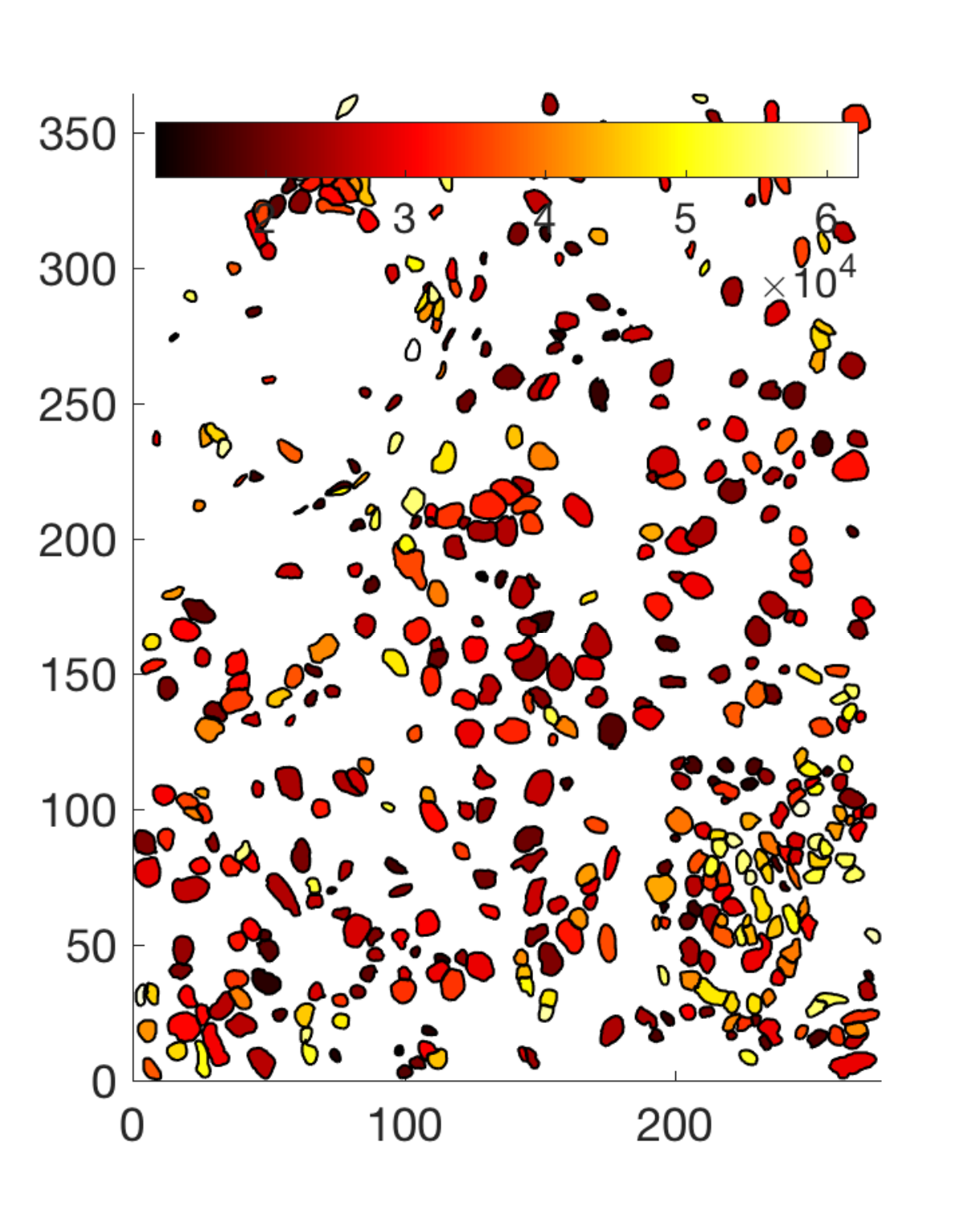}
    \put(15,95){\footnotesize Observed intensities}
       \put(-2,95){\footnotesize a.)}
    \put(37,2){\tiny $x$ ($\mu\text{m}$)}
    \put(-1,42){\tiny \rotatebox{90}{$y$ ($\mu\text{m}$)}}
    \put(-10,41){\footnotesize \rotatebox{90}{DAPI}}
    \end{overpic} %
        \begin{overpic}[width=0.23\textwidth]{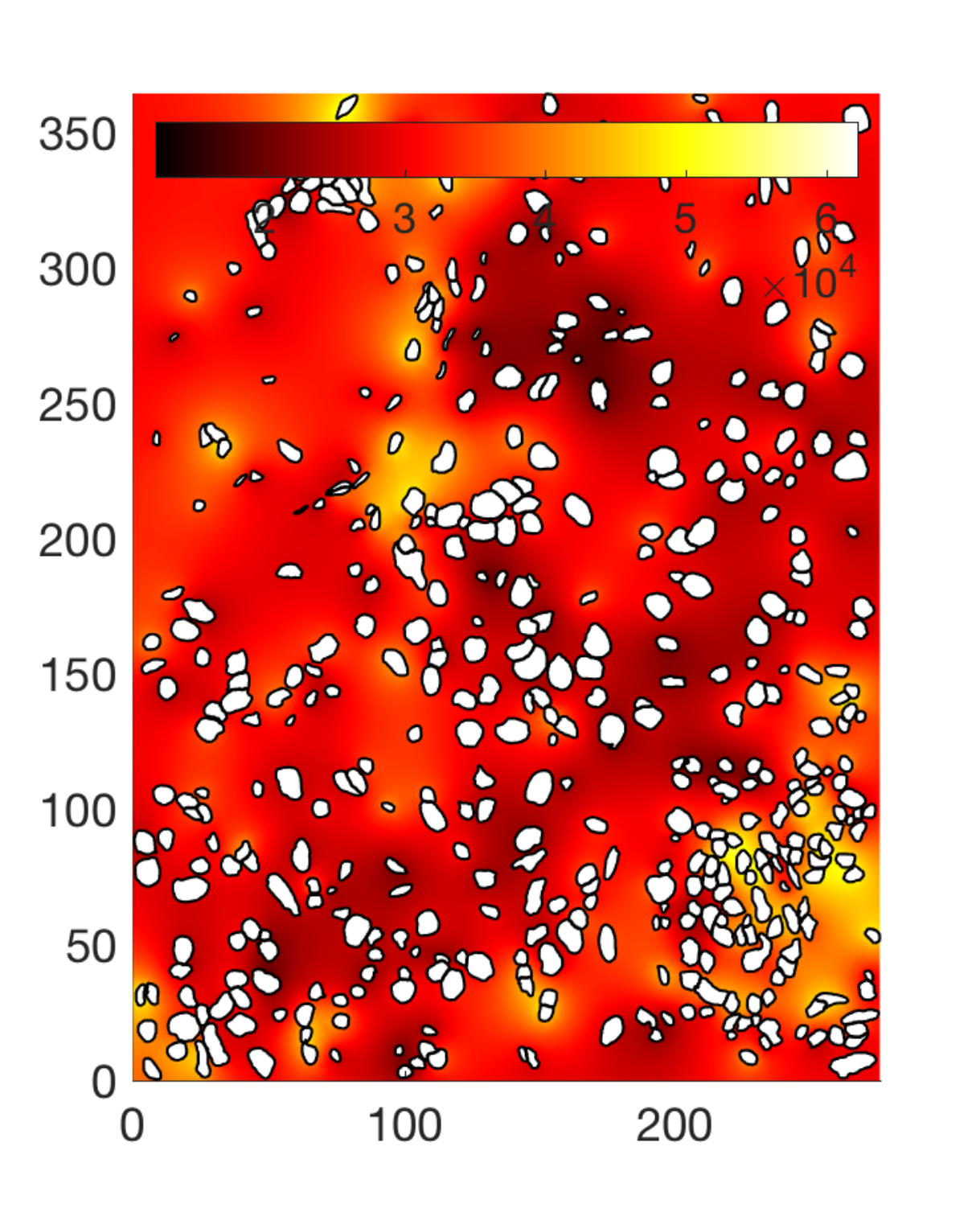}
    \put(22,95){\footnotesize Signalling Field}
       \put(-2,95){\footnotesize b.)}
    \put(37,2){\tiny $x$ ($\mu\text{m}$)}
    \put(-1,42){\tiny \rotatebox{90}{$y$ ($\mu\text{m}$)}}
    \end{overpic}
    \begin{overpic}[width=0.23\textwidth]{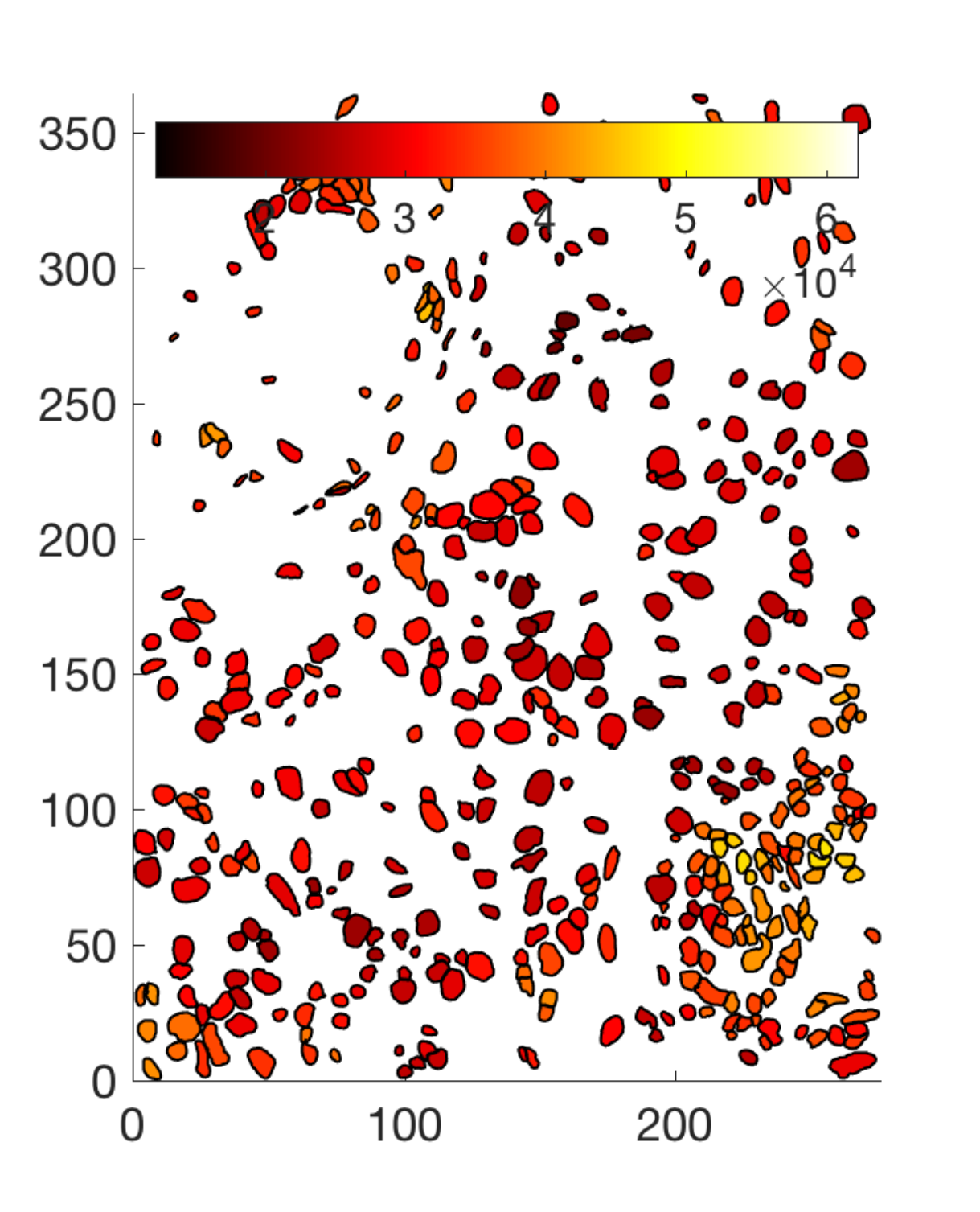}
        \put(17,95){\footnotesize Baseline intensities}
        \put(-2,95){\footnotesize c.)}
    \put(37,2){\tiny $x$ ($\mu\text{m}$)}
    \put(-1,42){\tiny \rotatebox{90}{$y$ ($\mu\text{m}$)}}
    \end{overpic}  
    \begin{overpic}[width=0.23\textwidth]{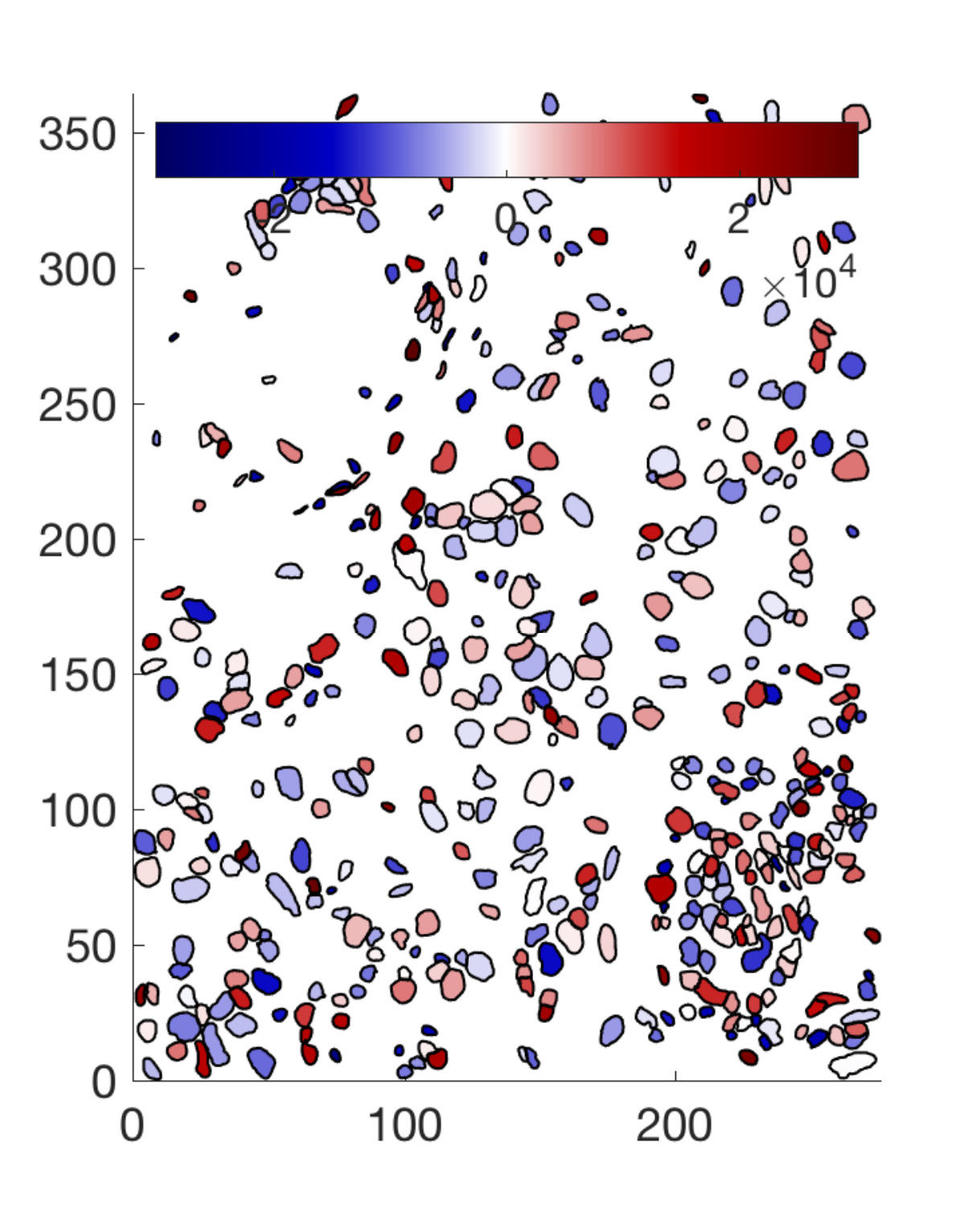}
     \put(27,95){\footnotesize Cell impact}
      \put(-2,95){\footnotesize d.)}
    \put(37,2){\tiny $x$ ($\mu\text{m}$)}
    \put(-1,42){\tiny \rotatebox{90}{$y$ ($\mu\text{m}$)}}
    \end{overpic} 
    \\
    
        \begin{overpic}[width=0.23\textwidth]{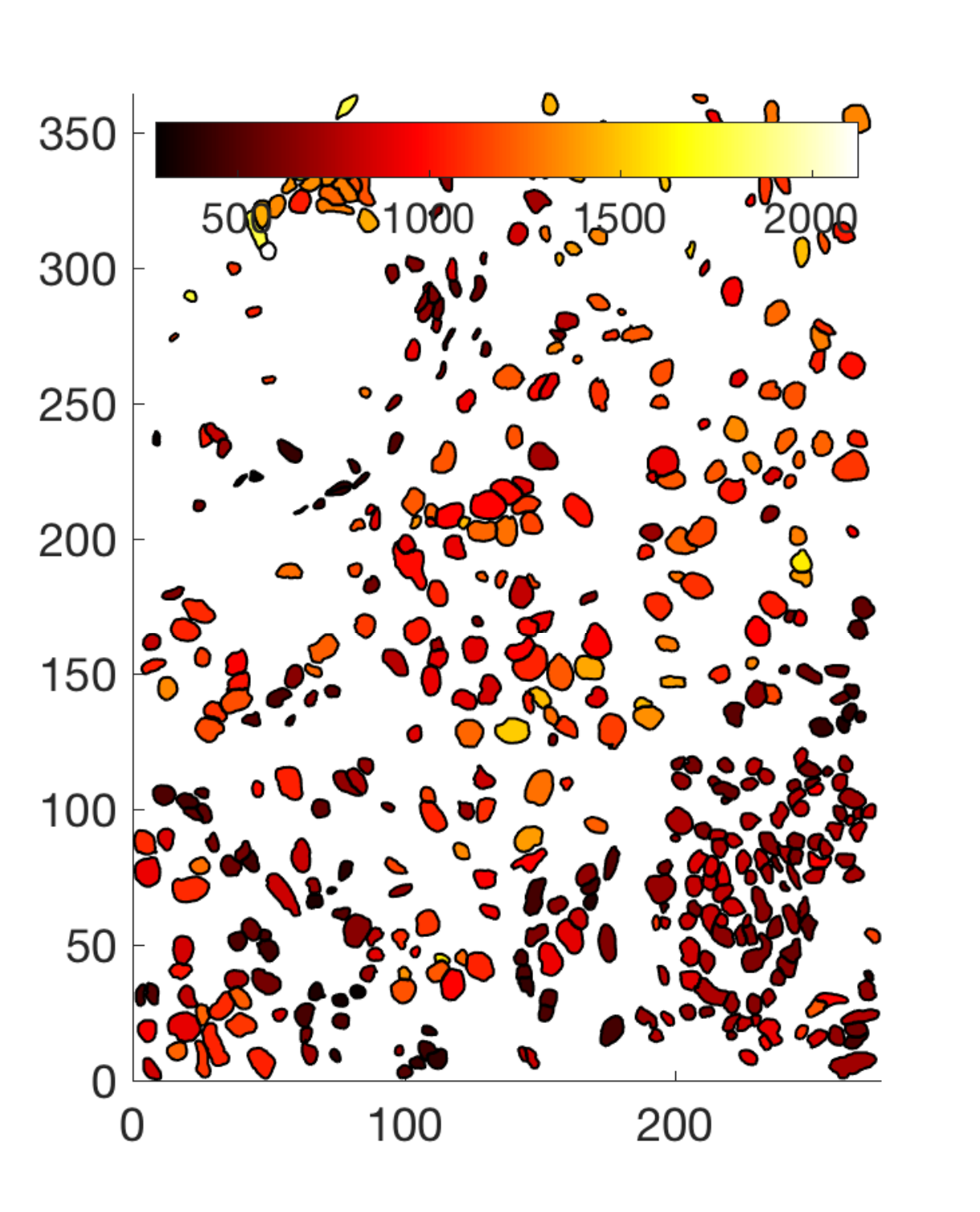}
       \put(-2,95){\footnotesize e.)}
    \put(37,2){\tiny $x$ ($\mu\text{m}$)}
    \put(-1,42){\tiny \rotatebox{90}{$y$ ($\mu\text{m}$)}}
     \put(-10,41){\footnotesize \rotatebox{90}{HER2}}
    \end{overpic} %
        \begin{overpic}[width=0.23\textwidth]{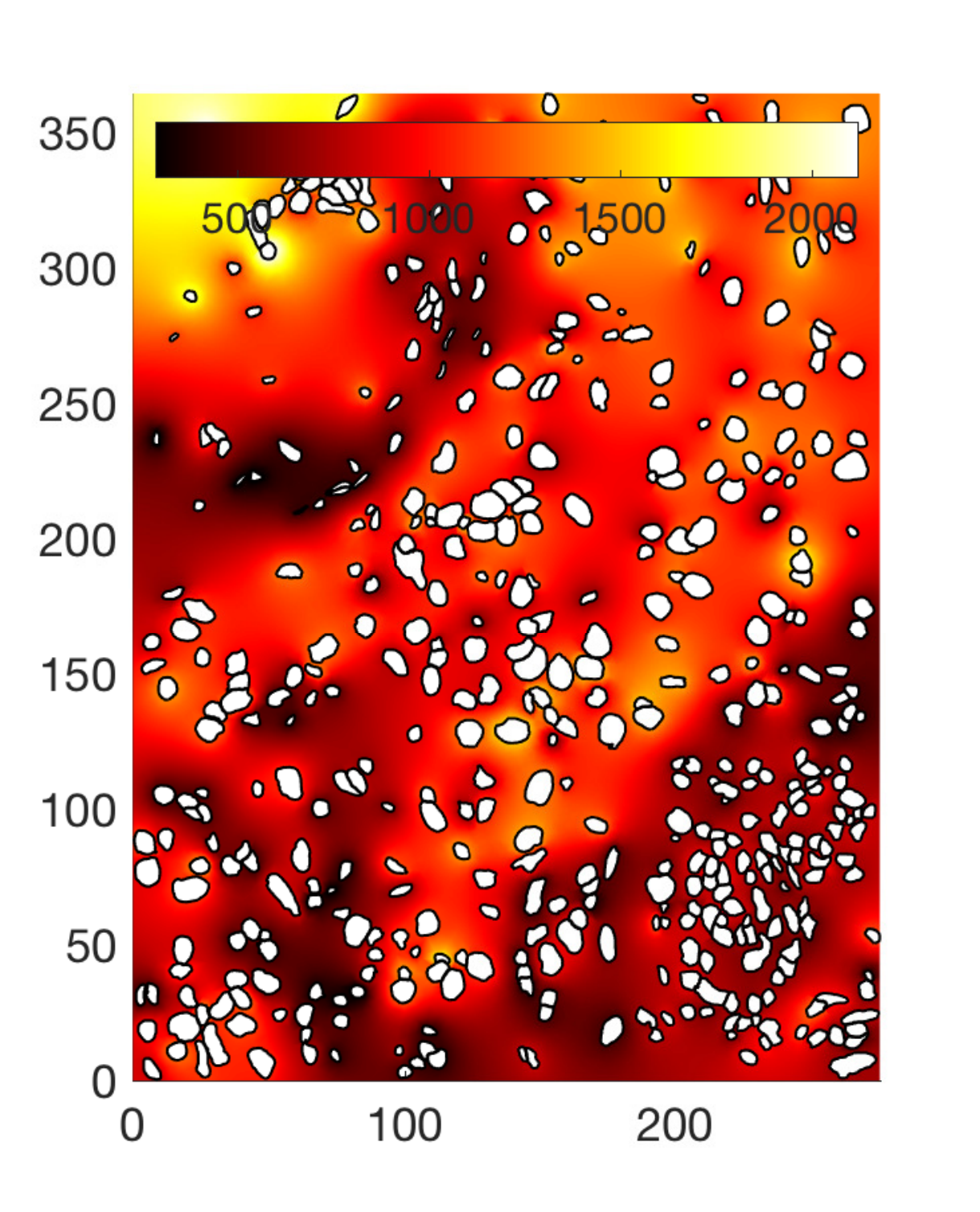}
       \put(-2,95){\footnotesize f.)}
    \put(37,2){\tiny $x$ ($\mu\text{m}$)}
    \put(-1,42){\tiny \rotatebox{90}{$y$ ($\mu\text{m}$)}}
    \end{overpic} 
    \begin{overpic}[width=0.23\textwidth]{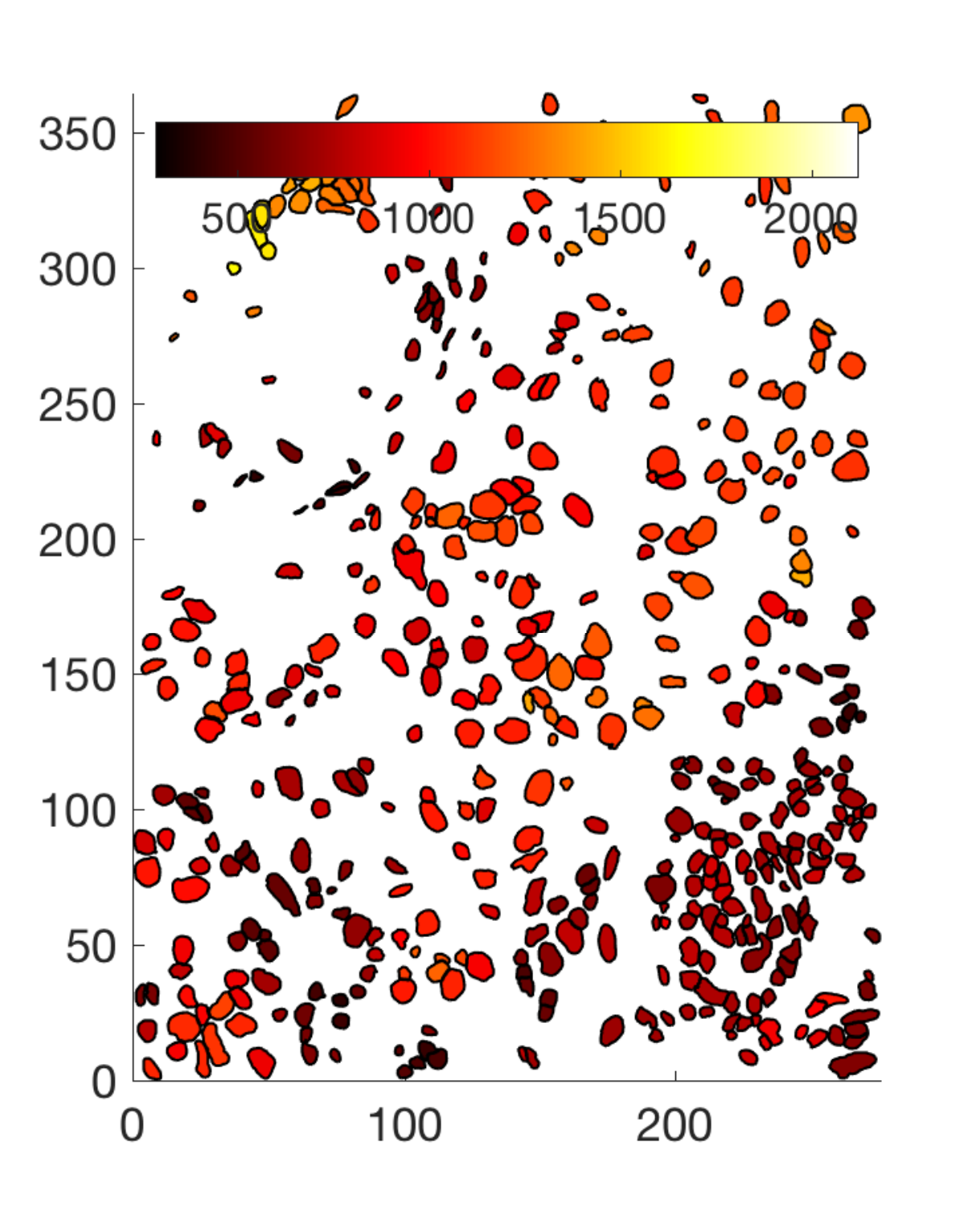}
        \put(-2,95){\footnotesize g.)}
    \put(37,2){\tiny $x$ ($\mu\text{m}$)}
    \put(-1,42){\tiny \rotatebox{90}{$y$ ($\mu\text{m}$)}}
    \end{overpic}  
    \begin{overpic}[width=0.23\textwidth]{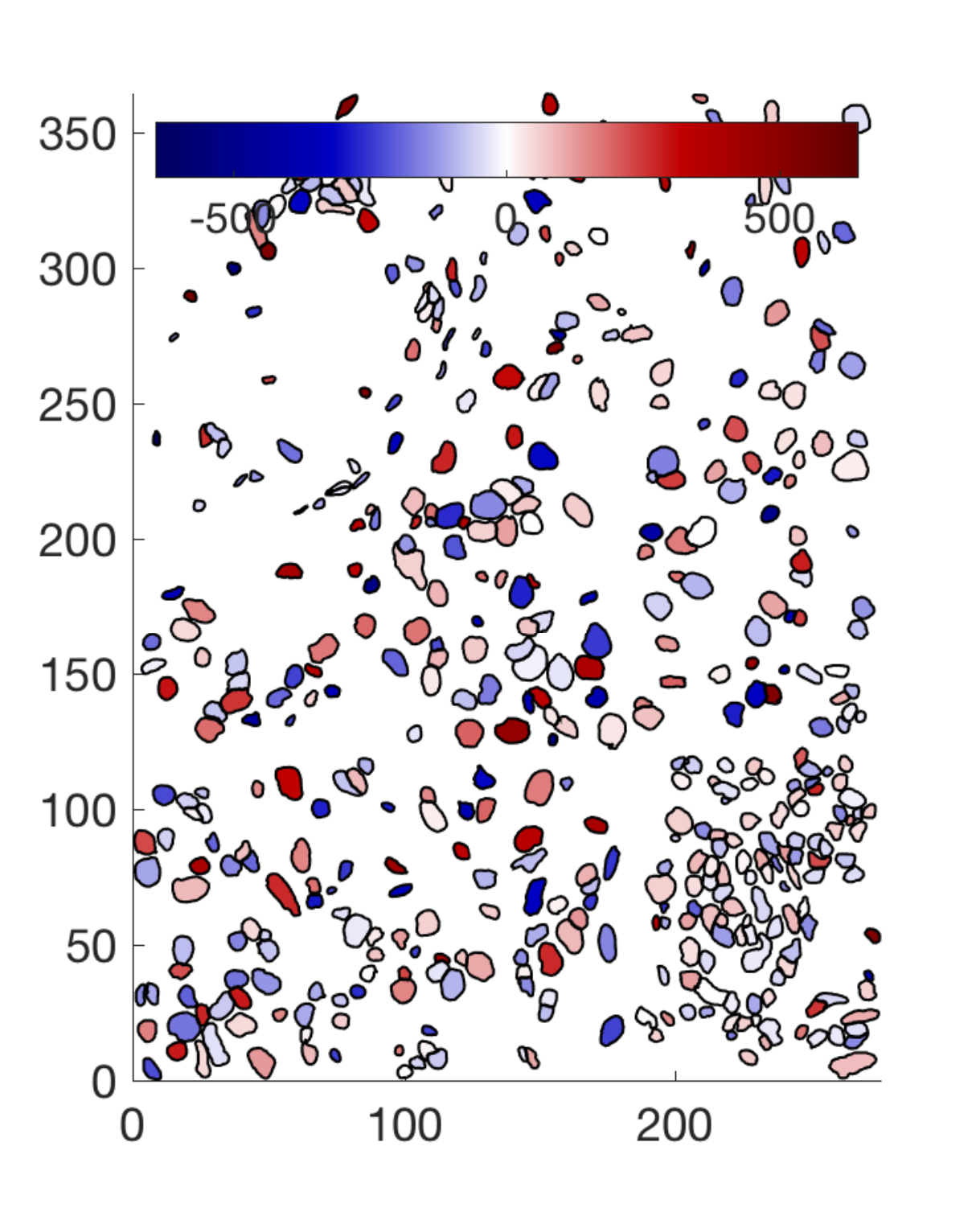}
      \put(-2,95){\footnotesize h.)}
    \put(37,2){\tiny $x$ ($\mu\text{m}$)}
    \put(-1,42){\tiny \rotatebox{90}{$y$ ($\mu\text{m}$)}}
    \end{overpic} 
 \\
    
      \begin{overpic}[width=0.23\textwidth]{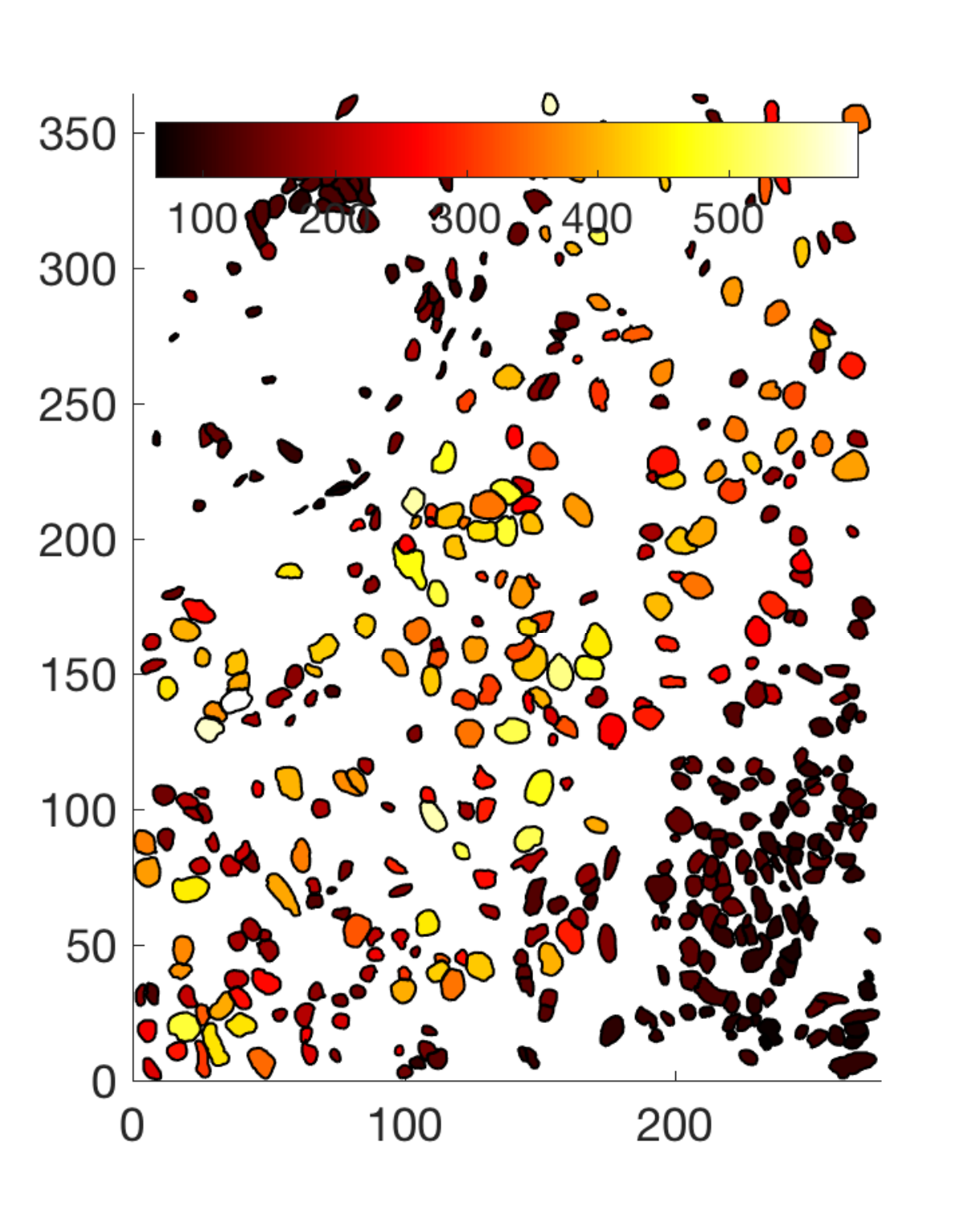}
        \put(-2,95){\footnotesize i.)}
    \put(37,2){\tiny $x$ ($\mu\text{m}$)}
    \put(-1,42){\tiny \rotatebox{90}{$y$ ($\mu\text{m}$)}}
     \put(-10,46){\footnotesize \rotatebox{90}{ER}}
    \end{overpic} %
        \begin{overpic}[width=0.23\textwidth]{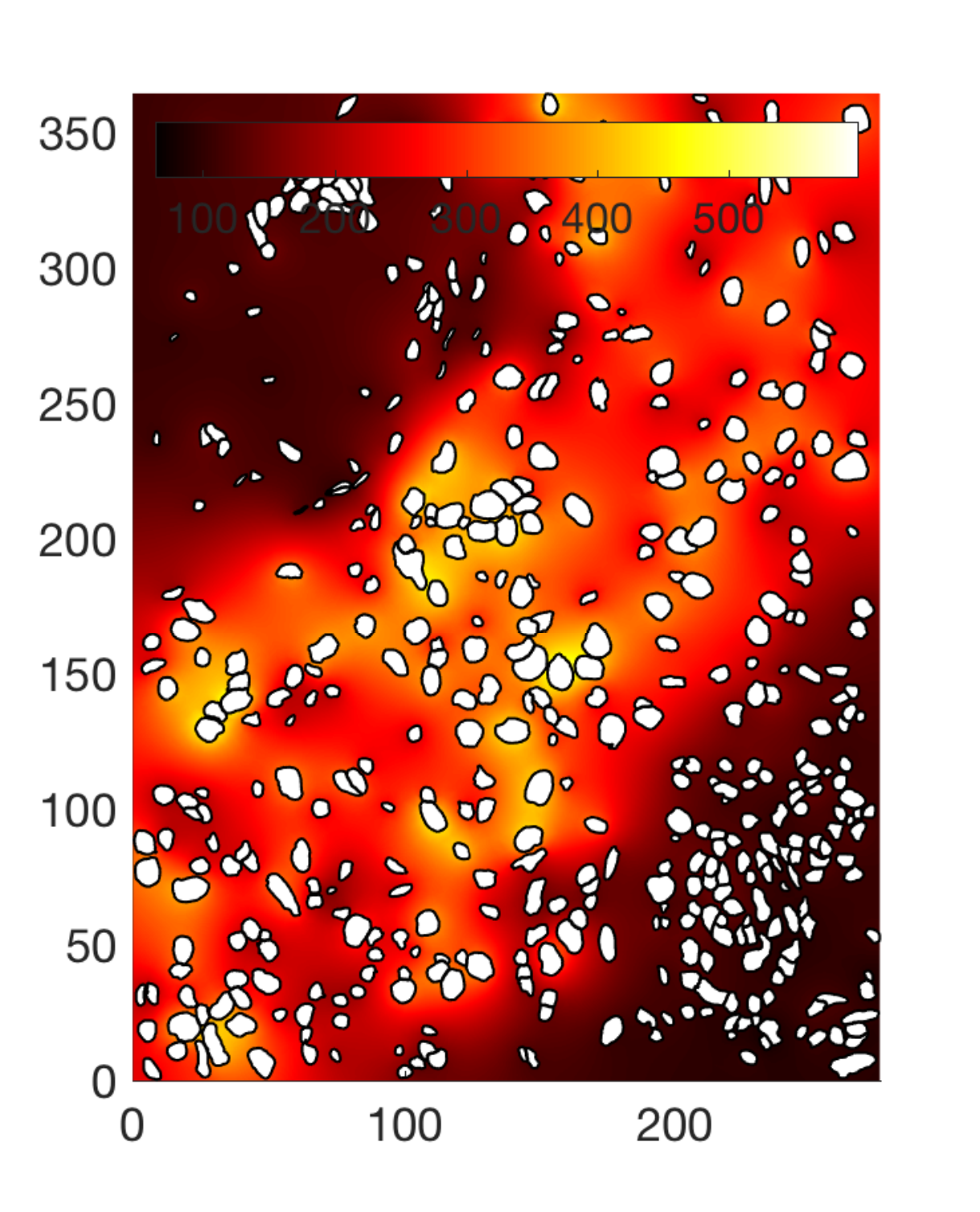}
       \put(-2,95){\footnotesize j.)}
    \put(37,2){\tiny $x$ ($\mu\text{m}$)}
    \put(-1,42){\tiny \rotatebox{90}{$y$ ($\mu\text{m}$)}}
    \end{overpic} 
    \begin{overpic}[width=0.23\textwidth]{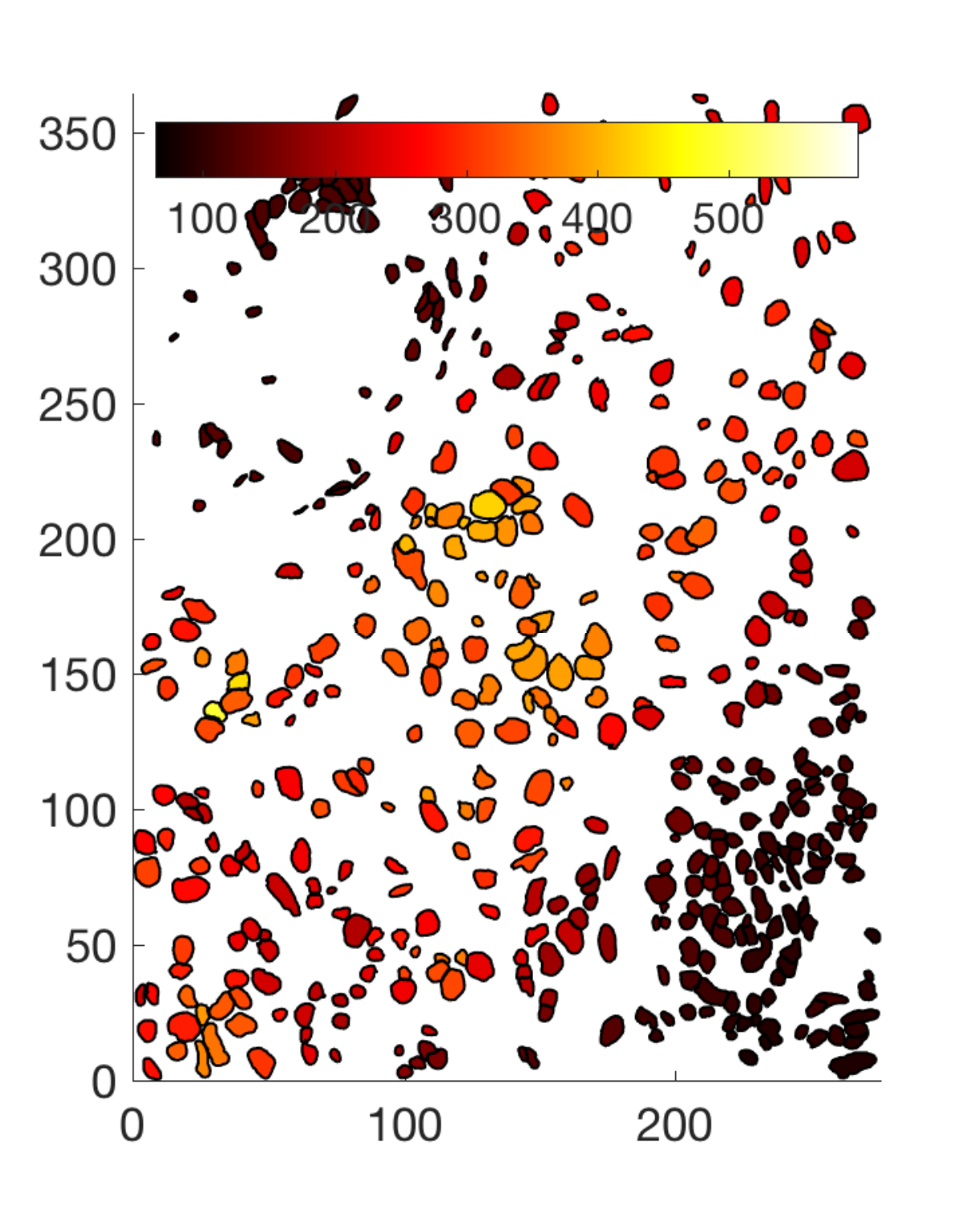}
        \put(-2,95){\footnotesize k.)}
    \put(37,2){\tiny $x$ ($\mu\text{m}$)}
    \put(-1,42){\tiny \rotatebox{90}{$y$ ($\mu\text{m}$)}}
    \end{overpic}  
    \begin{overpic}[width=0.23\textwidth]{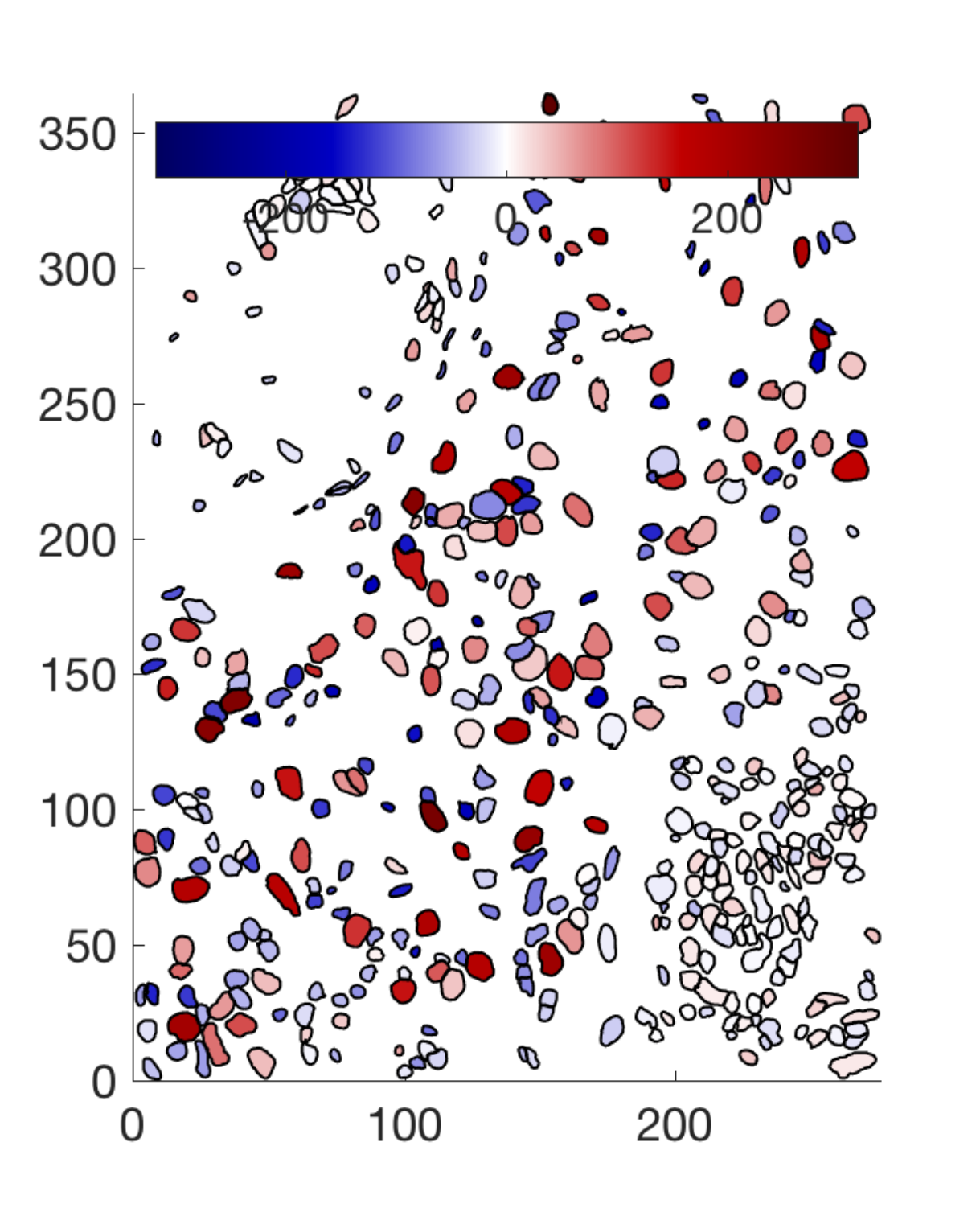}
      \put(-2,95){\footnotesize l.)}
    \put(37,2){\tiny $x$ ($\mu\text{m}$)}
    \put(-1,42){\tiny \rotatebox{90}{$y$ ($\mu\text{m}$)}}
    \end{overpic} 

    \caption{\footnotesize{Pathology slide with outlines plotted. Figures showing stain intensity for: (left column) \emph{observed} intensities $\cvec$; (left centre column) the exterior signalling field --- the solution to equation \eqref{eq_gen_chem_prof}; (right centre column) \emph{baseline} intensities $\bvec$; (right column) \emph{cell impact} $\fvec$ --- for this column red colours correspond to secreting cells, blue colours correspond to absorbing cells. The rows correspond to: (top row) \DAPI stain; (middle row) \HER stain; and (bottom row) \ER stain.}}
  \label{fig_breast_spatial_staining_patterns}
\end{figure}
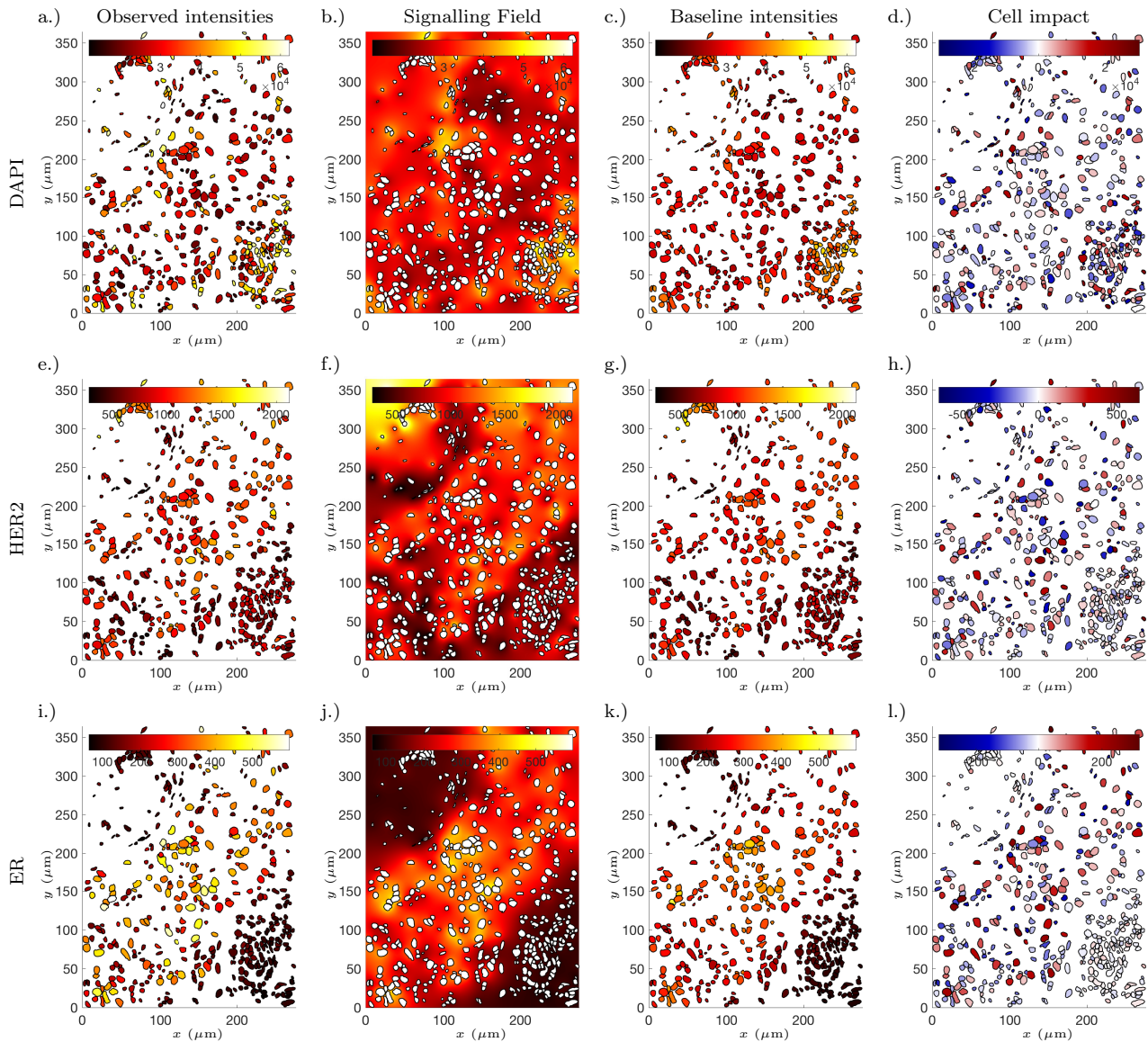

\subsection{Possible Identification of Microenvironmental Niches}

For Sample 1, ignoring the \DAPI stain, we plot the \emph{baseline} \HER and \emph{baseline} \ER data sets as a scatter graph in Figure \ref{fig_breast_types}(a). One observes approximately 3 clusters of points (which can be found using $k$-means clustering). When looking at the \emph{observed} staining intensities, these clusters of points appear as less distinct entities, see Figure \ref{fig_breast_types}(b). An interpretation of this finding is that there are 3 distinct microenvironmental niches on the slide. Viewing these clusters spatially, one can see that cells with approximately the same microenvironment are located in adjacent locations, see Figure \ref{fig_breast_types}(c). Samples 2 and 3 seemed to also have niches, but without such immediately identifiable groups.

\begin{SCfigure}
\centering
     \begin{overpic}[width=0.28\textwidth]{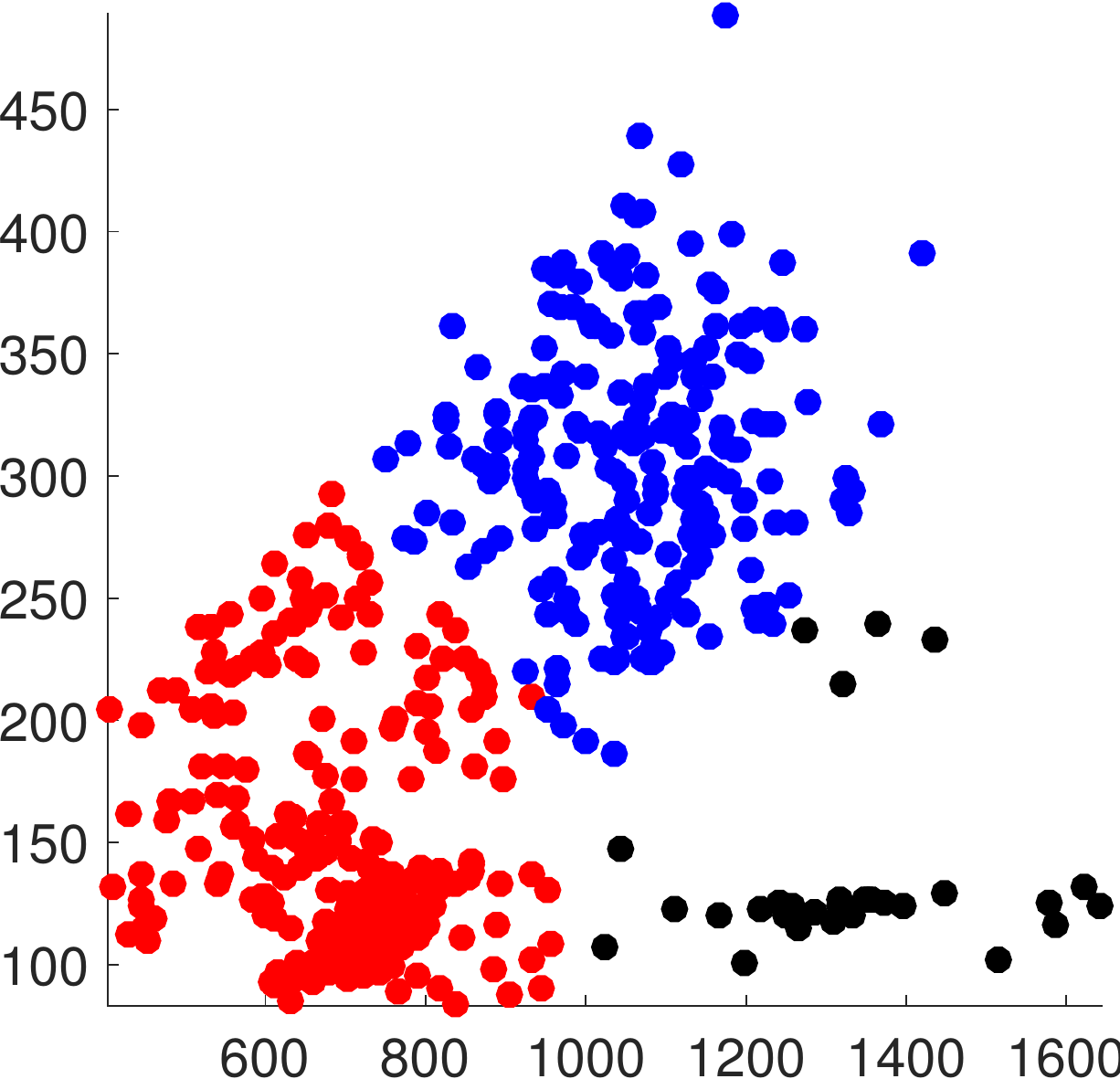}
       \put(-2,98){\footnotesize a.)}
       \put(27,98){\footnotesize Baseline intensities}
    \put(48,-6){\footnotesize HER2}
    \put(-6,48){\footnotesize \rotatebox{90}{ER}}
    \end{overpic} %
    \hspace{.2 cm}
     \begin{overpic}[width=0.28\textwidth]{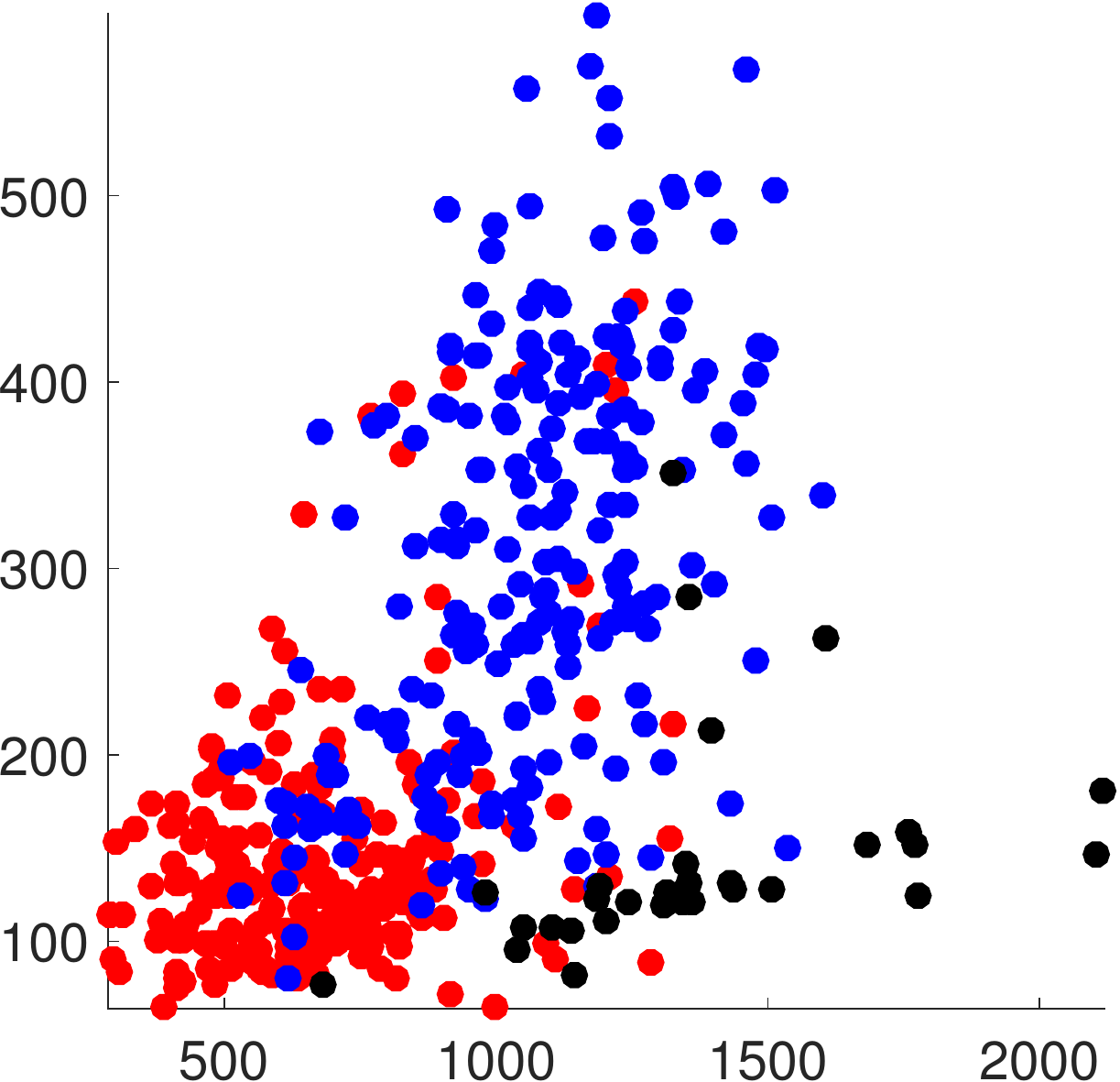}
       \put(-2,98){\footnotesize b.)}
       \put(25,98){\footnotesize Observed intensities}
    \put(48,-6){\footnotesize HER2}
    \put(-6,48){\footnotesize \rotatebox{90}{ER}}
    \end{overpic} 
     \begin{overpic}[width=0.24\textwidth]{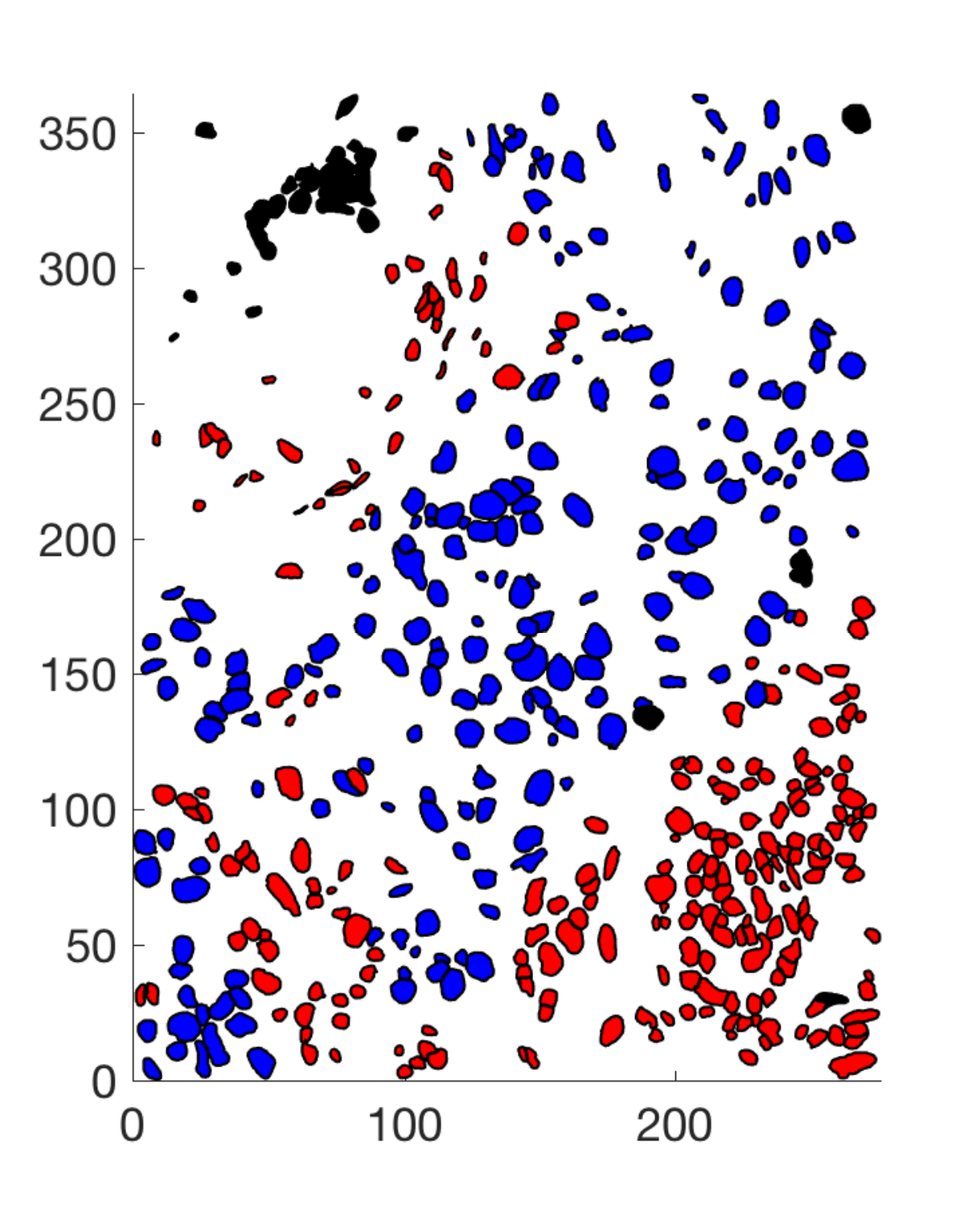}
       \put(-4,93){\footnotesize c.)}
       \put(34,2){\footnotesize $x$ ($\mu\text{m}$)}
    \put(-3,42){\footnotesize \rotatebox{90}{$y$ ($\mu\text{m}$)}}
    \end{overpic} 
    \vspace{6mm}
\caption{\footnotesize{Figure to show possible existence of microenvironmental niches. a.) $k$-means clustering applied to \HER and \ER \emph{baseline} intensities; b.) The labelling from the $k$-means clustering shown on the \HER and \ER \emph{observed} data sets; and c.) the $k$-means clustering labels applied spatially.}}
  \label{fig_breast_types}
\end{SCfigure}
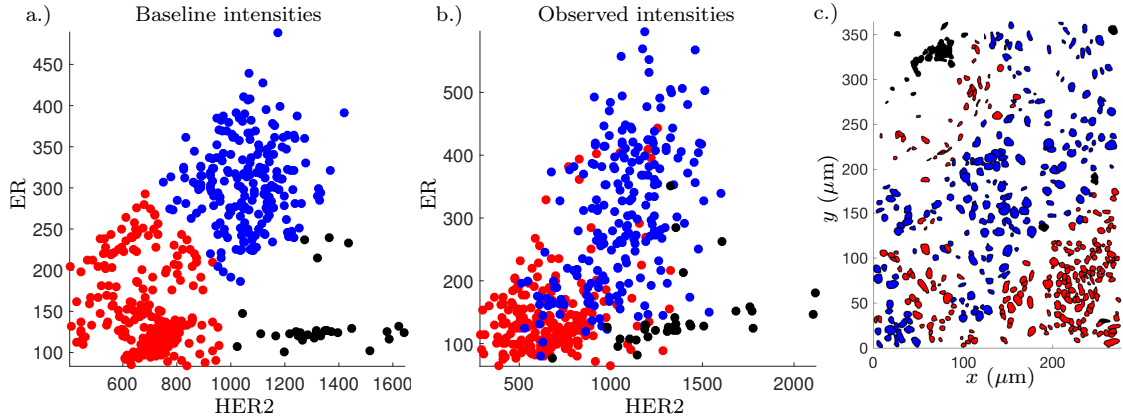

\subsection{Use of Centromeric Probes to Investigate \HER Gene Amplification Status}

Using the FISH experimental procedure, a \HER gene specific probe was used in conjunction with a centromeric probe to identify the copy number of the \HER gene for each cell. We write the copy number of cell $i$ as
\begin{equation}
\cn_i = \frac{\text{BAC}}{\text{CEP}} = \frac{\text{\# Number of \HER gene probes detected}}{\text{\# Number of centromeres detected}} \, .
\end{equation}
As the data is noisy, we discard data points that are biologically unreasonable, i.e., correspond to a amplification ratio that is infinite ($\cn_i = \infty$), or less than a single specific probe detected per centromere ($\cn_i < 1$). We then put the data into two ordinal sets, either unamplified wildtype (WT) cells where the \HER gene has not been upregulated ($\cn_i =  1$), or where the cell has amplified the \HER gene ($\cn_i > 1$). For the 3 data sets, after invalid data has been discarded, amplified \HER cells account for $26\%$, $9\%$ and $17\%$ (respectively) of the cells on the slide. 

For the first data set, using a logistic regression, one can use the measured \HER staining intensities $\cvec$ to predict whether the \HER gene is amplified $\cnvec = \{\cn_1,\dots,\cn_N\}$. This logistic regression is correct with $66\%$ accuracy. Analysing the binary classified data sets for staining intensity (either low or high staining intensity with the threshold determined by the logistic regression), we find that for high staining intensity cells tend to be net secretors of the SF ($\fv_i > 0$), and cells with low staining intensity tend to be net absorbers of the SF ($\fv_i < 0$). 

However, when looking at cells with low \HER staining intensity, but with amplified \HER gene, we find that the \emph{cell impact} for the \HER stain is greater than the mean for all low \HER stain intensity cells (see Figure \ref{fig_BCA_Hists}). By use of $z$-tests (used for $n>30$ samples), we reject the null hypothesis that: \emph{the mean \emph{cell impact} for cells with low \HER stain intensity and amplified \HER gene is the same as the population mean \emph{cell impact} for all cells with low \HER stain intensity}; compared to the alternative hypothesis that the means are not the same. This is found to be statistically significant with $p<0.05$. The interpretation of this is that when cells have a low staining intensity, but with an amplified \HER gene, the cells may not necessarily stain brightly but are net secretors into the SF, (and therefore they likely stain \emph{comparatively} more intensely than their immediate neighbours). 

Additionally, it should be noted for cells with low staining intensity: for amplified \HER copy number, cells are net producers of the SF impacting \ER intensity; and for unamplified \HER cells, cells are net absorbers of the SF impacting \ER intensity (for both cases using $z$-tests, $p<0.05$).

\begin{SCfigure}
  \centering
  \hspace{0.035\textwidth}
    \begin{overpic}[width=0.725\textwidth]{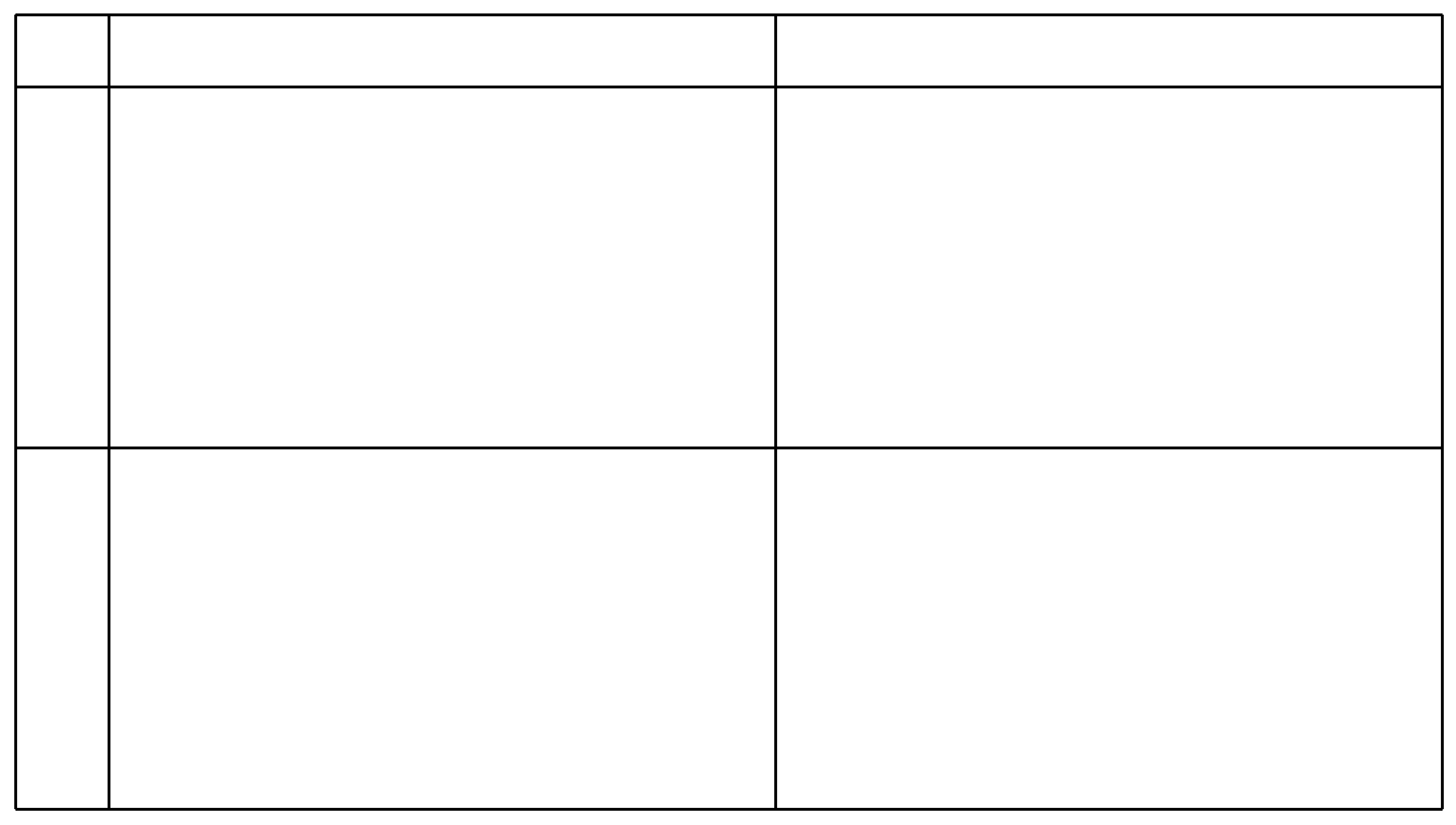}
    \put(-4,15){\Large \rotatebox{90}{Copy Number}}
       \put(3.5,6){\rotatebox{90}{Unamplified}}
       \put(3.5,33){\rotatebox{90}{Amplified}}
    \put(35,58){\Large \HER Stain Intensity}
     \put(28,52.5){Low}
      \put(75,52.5){High}
    \put(14,2){\includegraphics[width=0.26\textwidth]{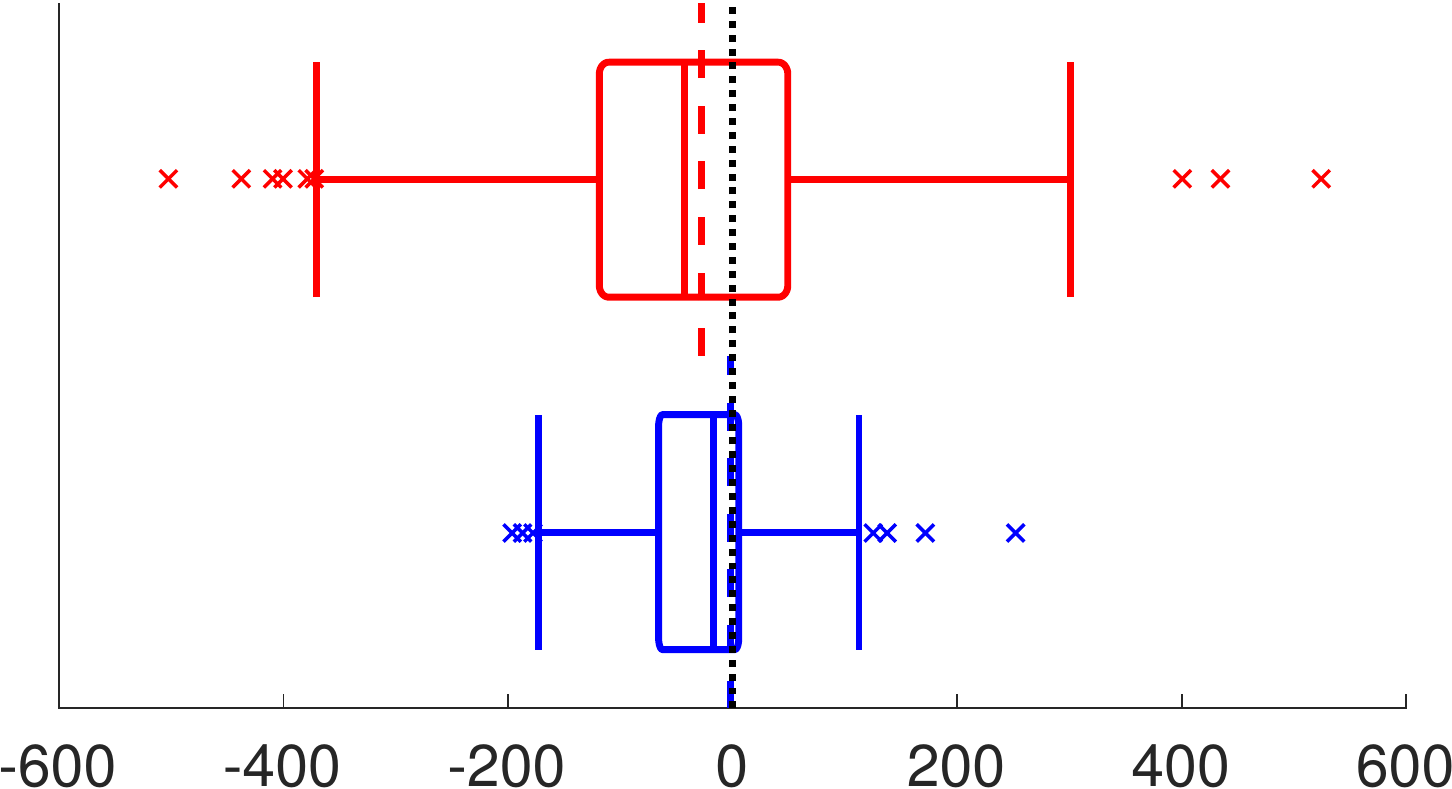}}
         \put(40,23){\small $n = 138$}
         \put(12,13.75){\small \rotatebox{90}{\HER}}
         \put(12,6.75){\small \rotatebox{90}{\ER}}
         \put(49,7.75){($*$)}
    \put(60,2){\includegraphics[width=0.26\textwidth]{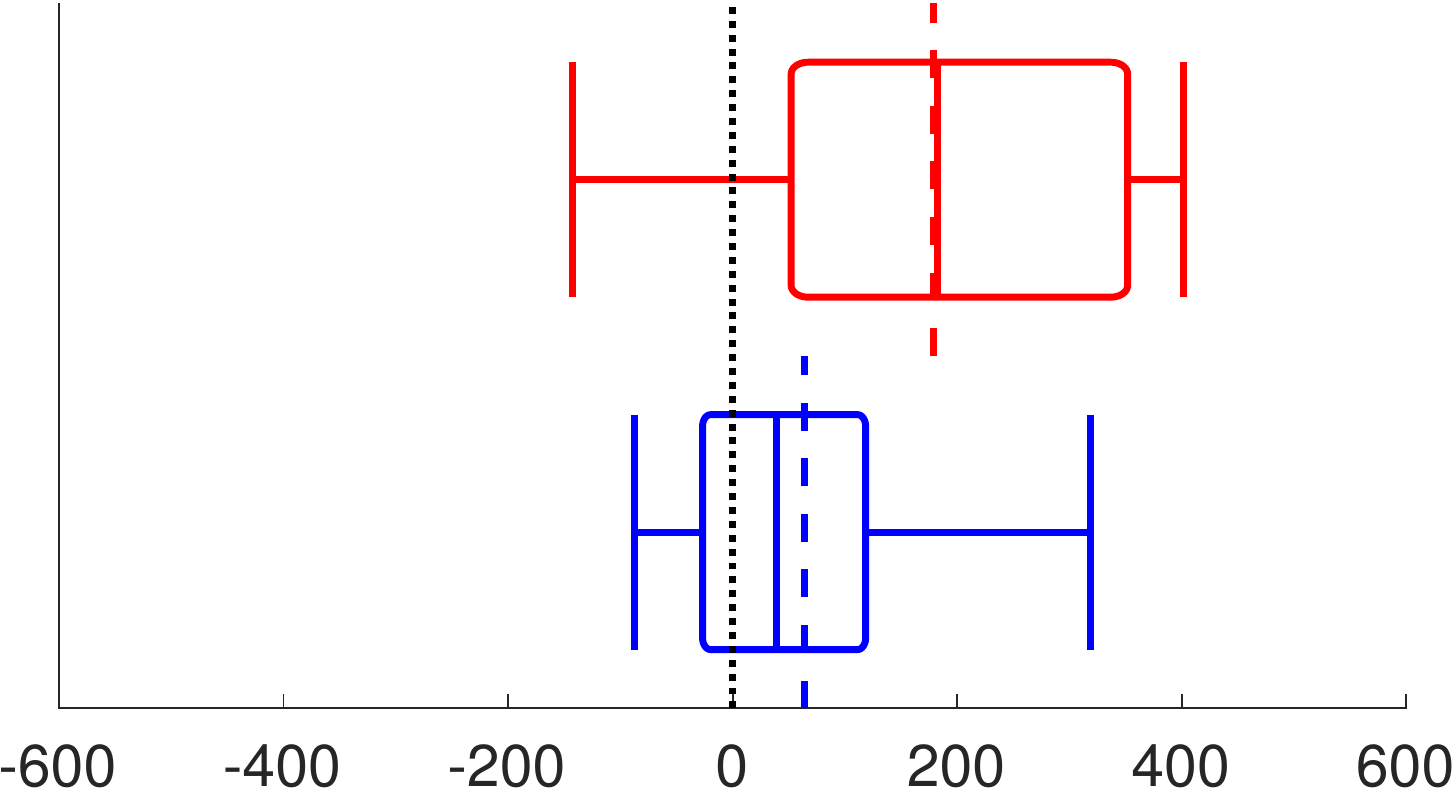}}
          \put(86,23){\small $n = 16$}
          \put(58,13.75){\small \rotatebox{90}{\HER}}
         \put(58,6.75){\small \rotatebox{90}{\ER}}
    \put(14,27){\includegraphics[width=0.26\textwidth]{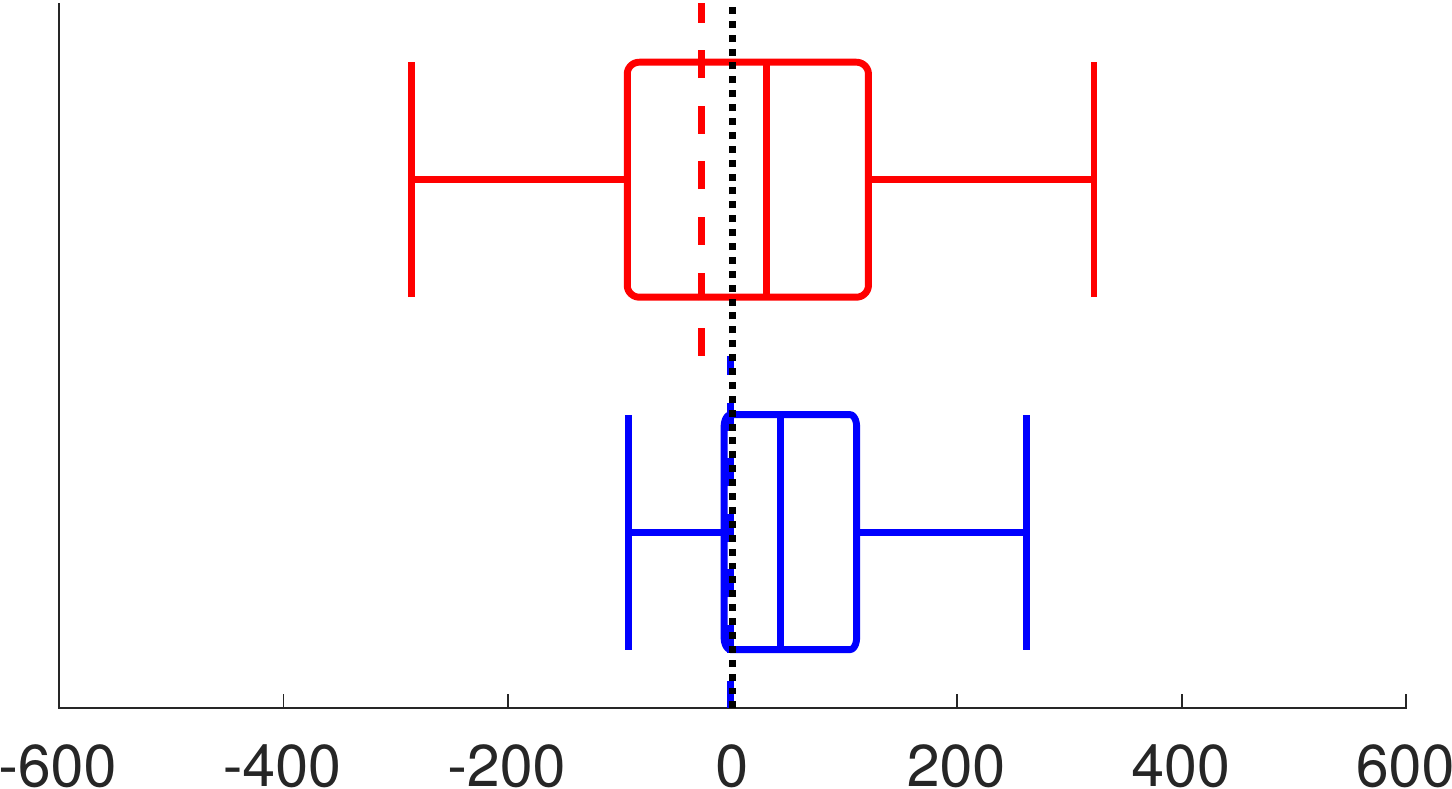}}
    	 \put(40,48){\small $n = 63$}
         \put(12,38.75){\small \rotatebox{90}{\HER}}
         \put(12,31.75){\small \rotatebox{90}{\ER}}
                  \put(49,41.75){($*$)}
                  \put(49,32.75){($*$)}
    \put(60,27){\includegraphics[width=0.26\textwidth]{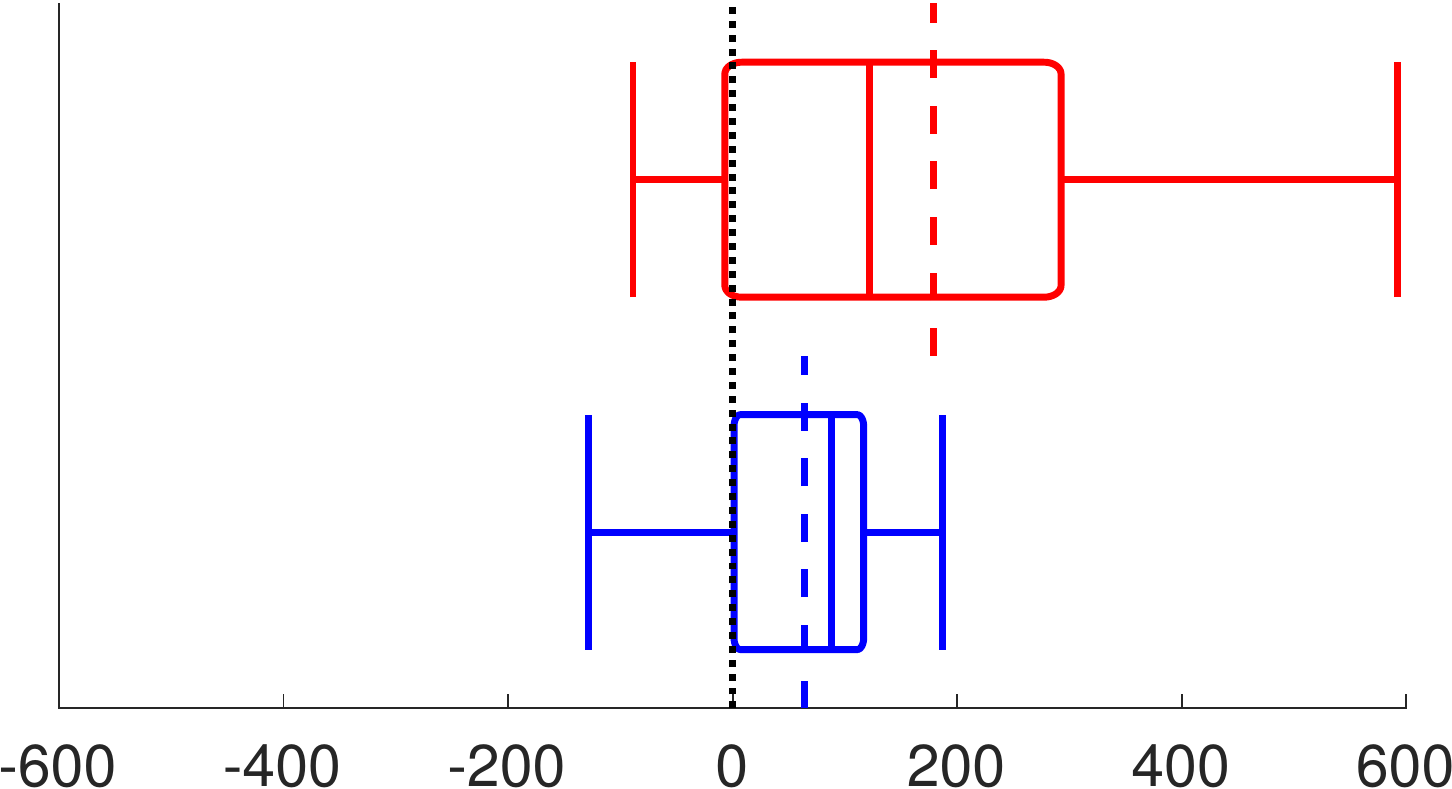}}
	  \put(86,48){\small $n = 16$} 
	           \put(58,38.75){\small \rotatebox{90}{\HER}}
         \put(58,31.75){\small \rotatebox{90}{\ER}}   
    \end{overpic} 
         
  \caption{\footnotesize{Turkey box plots for \emph{cell impact} data sets $\fvec$ for \HER (red) and \ER (blue) when divided by up by low/high \HER staining intensity ($\cvec$) and unamplified/amplified \HER gene expression ($\cnvec$). Vertical dashed lines show the mean values of the cell action ($\fvec$) for \HER and \ER with low and high staining intensities regardless of \HER gene copy number. Asterisks denote statistically significant differences between the data set denoted and the mean low/high \emph{HER2}/\ER cell action (vertical dashed line). The black dotted line corresponds to $\fv = 0$.}}
  \label{fig_BCA_Hists}
\end{SCfigure}
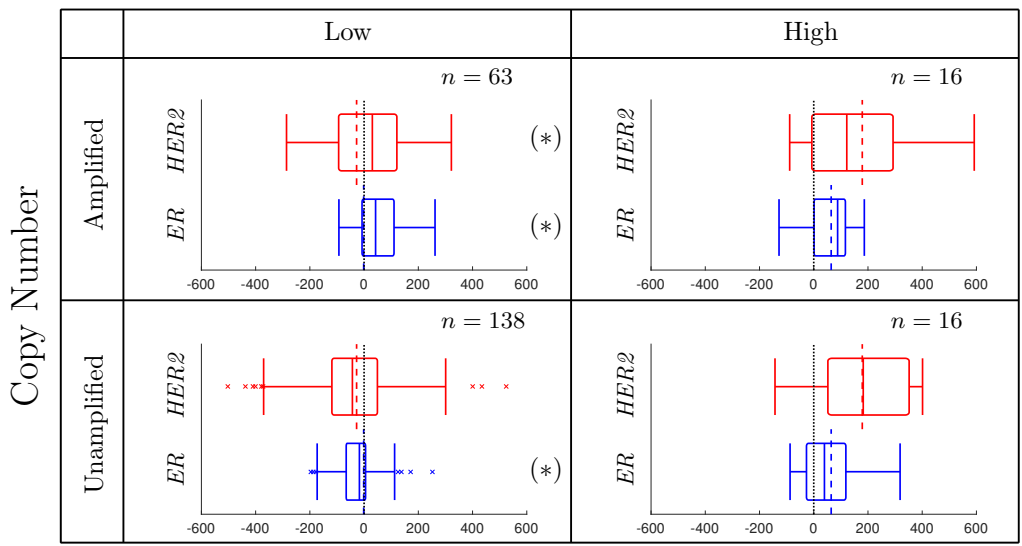

\section{Method Motivation via Toy System}\label{sec_toy_system}

In this section we evaluate the efficacy of our method on synthetic data (where we know the exact mathematical representation of the heterogeneity). We generate this data with the following simple model. For each cell $i$, we suppose that there is measurable quantity $y_i = y_i(t)$ corresponding to the staining intensity. We suppose that there are a number of cell phenotypes, labelled by $\tau$ in set $\mathcal{T}$. We suppose that a each cell produces a signalling field from other cells as well as to some internal dynamical system within the cell. Thus we write
\begin{equation}\label{eq_toy_y_single}
\frac{\dd y_i}{\dd t}  = \mu_i ( y_i) +  \phi_i ( y_i , u )  \, .
\end{equation}
Here $\mu_i$ represents the internal dynamics, and will account for the cell's phenotype $\tau$, and $\phi_i$ represents the forcing by the signalling field $u$. As before we suppose that the dynamics of $y_i$ are slower than the diffusion of the SF, so that $u$ is in quasi-steady state. In the method described in Section \ref{sec_method_overview}, this had the form
\begin{equation}\label{eq_phi_actual}
\phi_i ( y_i(t), u(\xv) ) = \frac{\beta}{\vert \CellDom_i \vert } \int_{ \p\CellDom_i  }  \left( u(\xv) - y_i(t) \right) \dd S \, ,
\end{equation}
where $\beta$ is a scaling constant with units length over time, and $u = u(\xv)$ obeys equations \eqref{eq_gen_chem_prof}--\eqref{eq_gen_chem_prof_bc}. The \emph{baseline} was then found by considering the average concentration at the cell's location were the cell not there.

For simplicity, in our synthetic system, we will consider cells occupying negligible area, i.e, they will behave as point sources/sinks of the SF --- and therefore we do not need to incorporate boundary conditions. In this case, the \emph{baseline} ($\bv_i$) can be calculated by removing the internal dynamics for cell $i$, i.e., by setting $\mu_i = 0$.

With cells as points, we take $\phi_i$ as the Green's function solution to equation \eqref{eq_gen_chem_prof}--\eqref{eq_gen_chem_prof_bc}. Thus,
\begin{equation}\label{eq_form_phi_i}
\phi_i =   \beta  \sum_{j\neq i} G_n (d_{ij},\alpha) \left[ y_j(t)  - y_i(t) \right] \, ,
\end{equation}
where $d_{ij}$ represents the distance between cells $i$ and $j$. The form of $G_n$ depends on the dimension, and is given by
\begin{equation}
G_n (r,\alpha) = \left\{ \begin{array}{ccc} 
K_0( \alpha r) / 2\pi & \text{if} &n=2\, , \\
e^{-\alpha r } / 4\pi r & \text{if} & n=3\, ,
\end{array}\right.
\end{equation}
where $K_0$ is the modified Bessel function of the second kind. The \emph{baseline} is then given by
\begin{equation}
\bv_i = \sum_{j\neq i} c_j \, G_n(d_{ij},\alpha)   \bigg/  \sum_{j\neq i}  G_n(d_{ij},\alpha) \, .
\end{equation}
Notice that as cells are points, there is no $\gamma$ parameter (though in some sense it has been replaced by $\beta$).

We specify that there are three cell phenotypes $\tau \in \mathcal{T}= \{1,2,3\}$. We take $N=3000$. $1500$ cells are placed uniformly at random in the unit circle; $1500$ cells are placed uniformly at random in the annulus between $r=1$ and $r=2$ (see Figure \ref{fig_toy_model}). Using $k$-means clustering, cells are assigned into 15 groups. Groups are then assigned a cell type label randomly so that there are 5 groups of each cell type. The internal dynamics of each cell type are chosen to be given by
\begin{equation}
\mu_i ( y_i ) = \left\{ \begin{array}{ccc}  \eps_i /4  - y_i & \text{ if } & \tau_i = 1\, , \\   0  & \text{ if } & \tau_i = 2\, , \\ 1 + \eps_i /4 - y_i & \text{ if } & \tau_i = 3\, , \end{array} \right. \, 
\end{equation}
and we use the 3 dimensional Green's function as the interaction term. Here $\eps_i$ is a random number drawn from $\eps_i\sim \mathcal{N}(0,1)$ and is the means by which we model heterogeneity within a fixed cell phenotype. We choose $\beta = 4\pi/N$ and $\alpha = 10$. We can interpret the cell phenotypes as: cell phenotype $\tau=1$ is a cell that averages the signal received by neighbours to settle at a steady state close to some intrinsic value $\eps_i/4$; cell phenotype $\tau=2$ is a passive cell that mimics nearby cells; and cell phenotype $\tau=3$ is a cell that averages the signal received by neighbours to settle at a steady state close to some intrinsic value $1 + \eps_i/4$. From random initial conditions in the unit interval ($y_i(t=0)\in [0,1]$), the system will reach a steady state. The functional forms chosen are such that there is considerable overlap in the distribution of steady states for each cell phenotype, so that it is not immediately clear which phenotype a cell belongs to given the steady state value $\cv_i$.

In Figures \ref{fig_toy_model}(a,e), we see a spatial plot of the cell centres and a histogram of the \emph{observed} intensities ($\cvec$); Figure \ref{fig_toy_model}(b,f) reveals the discrete cell phenotypes that are not immediately apparent when viewing Figures \ref{fig_toy_model}(a,e). To calculate the \emph{baseline} ($\bvec$), we assume the functional form of the interactions are as in equation \eqref{eq_form_phi_i}, but without specifying the value of $\alpha$, which is determined via a best fit as in Section \ref{subsec_param_sel}. In Figure \ref{fig_toy_model}(c,g), we show a spatial plot and histogram of the \emph{baseline} cell intensities; and in Figure \ref{fig_toy_model}(d, h) we show the \emph{cell impact} cell intensities. From Figure \ref{fig_toy_model}(f) to Figure Figure \ref{fig_toy_model}(g), we see the emergence of a trimodal distribution indicating the three cell phenotypes. The phenotype $\tau=2$ is easily identified in Figure \ref{fig_toy_model}(h) as the population having zero \emph{cell impact}: this is to be expected since these cells essentially copy what their neighbours are doing.

Over 1000 simulations, we find that $\alpha_*$ has a mean value of $11.49$ [95\% CI: (8.65, 14.46)], and $R^2_{\max}$ has a mean value of $79.6\%$ [95\% CI: (76.9\%, 82.3\%)]. Carrying out $k$-means clustering on $\cvec$, $\bvec$ and $\fvec$ using \emph{a priori} knowledge that there are 3 phenotypes present, one finds that on average one identifies the 3 groups with 70.2\% [using $\cvec$, 95\% CI: (49.9\%, 82.5\%)], 74.3\% [using $\bvec$, 95\% CI: (57.4\%, 87.3\%)], and 15.7\% [using $\fvec$, 95\% CI: (11.0\%, 20.1\%)] accuracy respectively. Therefore, one can obtain slightly higher accuracy by looking at the \emph{baseline} intensities rather than the stain intensities. Moreover, if we identify the passive $\tau=2$ cells using $\fv_i \approx 0$, see Figure \ref{fig_toy_model}(h), then the remaining cells can be clustered into 2 groups with approximately 100\% accuracy.

With $\beta = 2\pi/N$, the 2 dimensional Green's function performs with similar accuracy regarding $R^2_{\max}$ values and clustering analysis. However, the $\alpha_*$ value is misidentified with mean value $16.66$ [95\% CI: (12.87, 20.79)].

\begin{figure}[t]
  \centering
  \vspace{1em} 
  \begin{overpic}[width=0.24\textwidth]{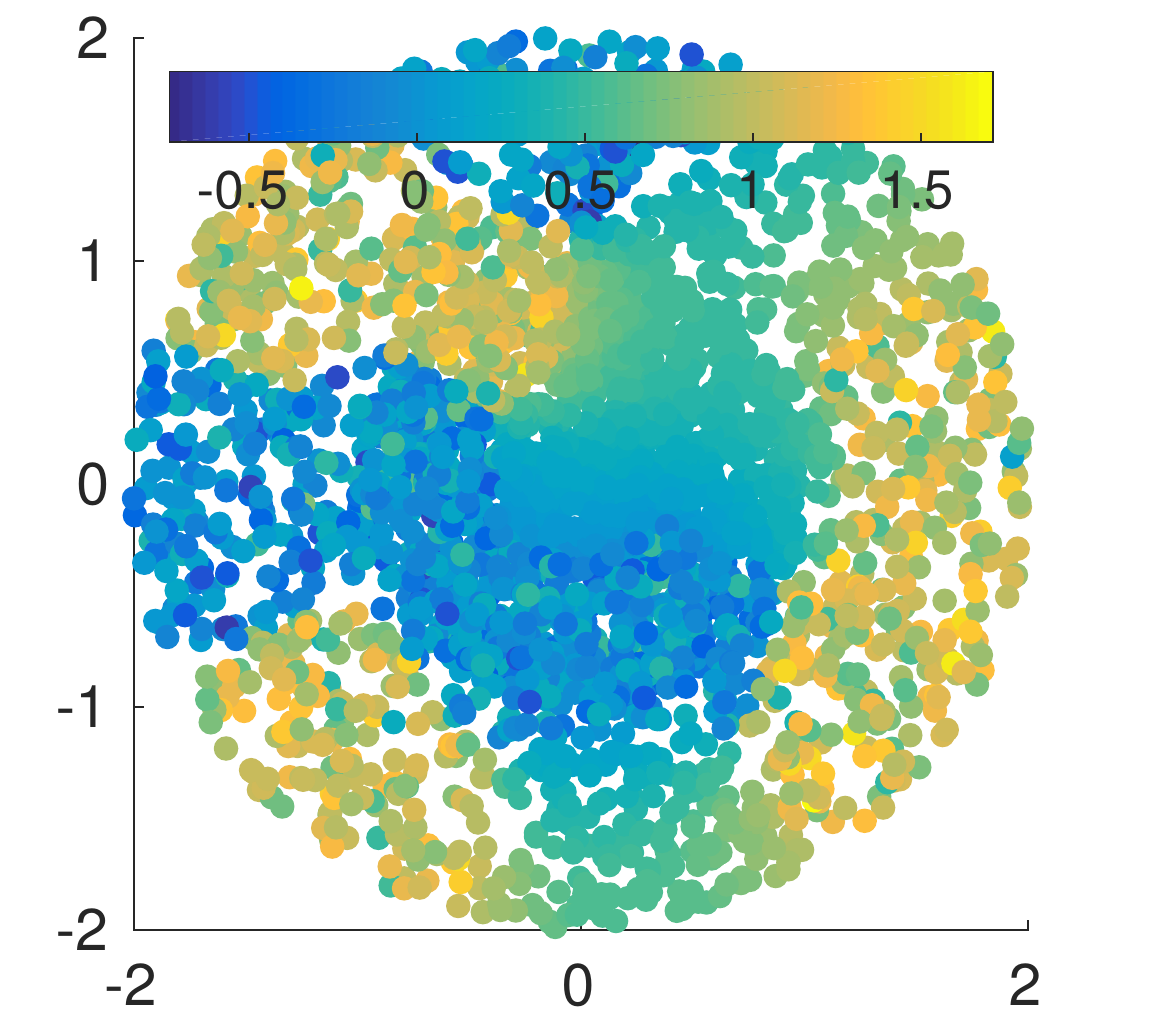}
  	\put(-3,83){\footnotesize a.)}
  \end{overpic} 
    \begin{overpic}[width=0.24\textwidth]{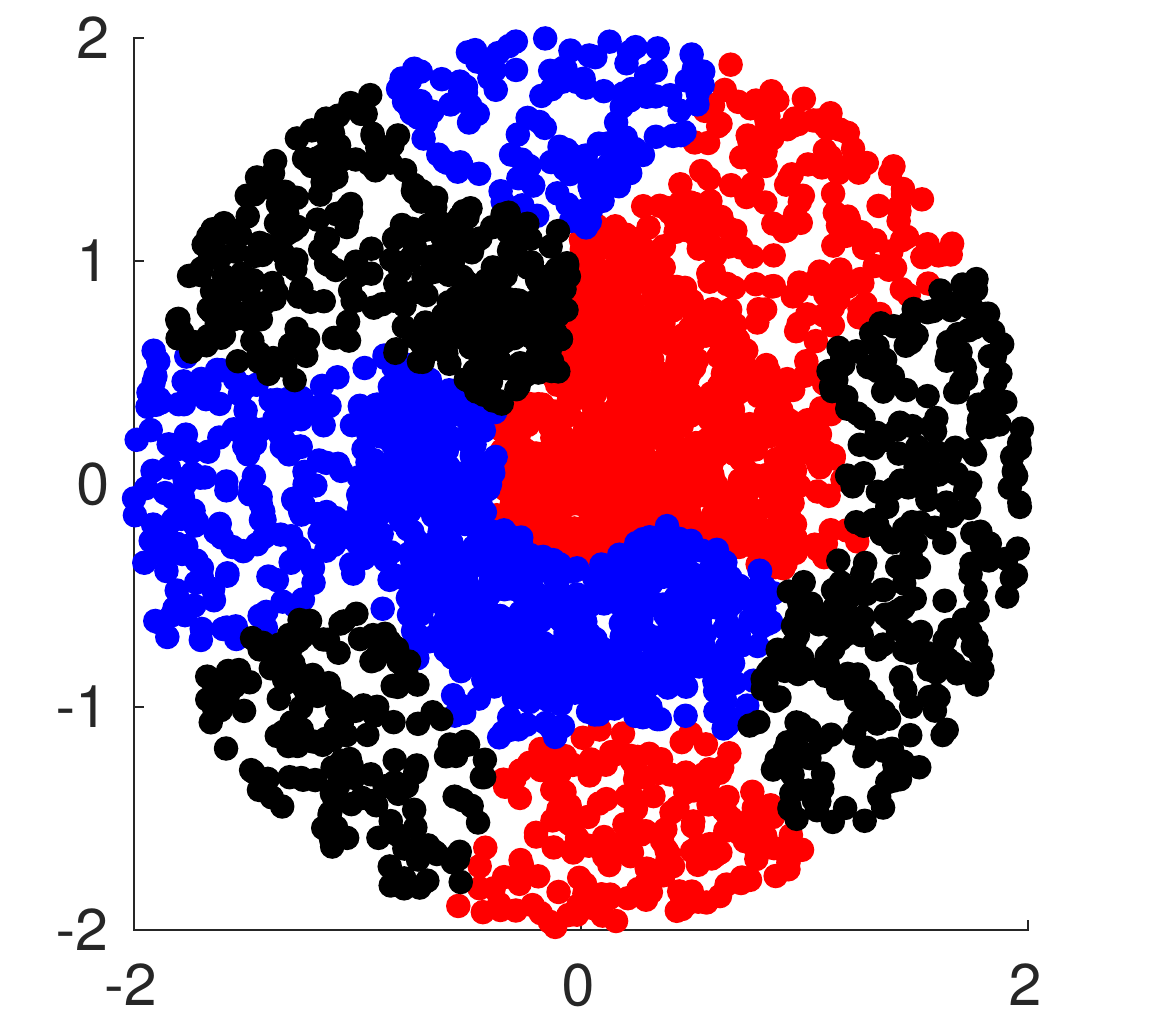}
  	\put(-3,83){\footnotesize b.)}
 \end{overpic} 
     \begin{overpic}[width=0.24\textwidth]{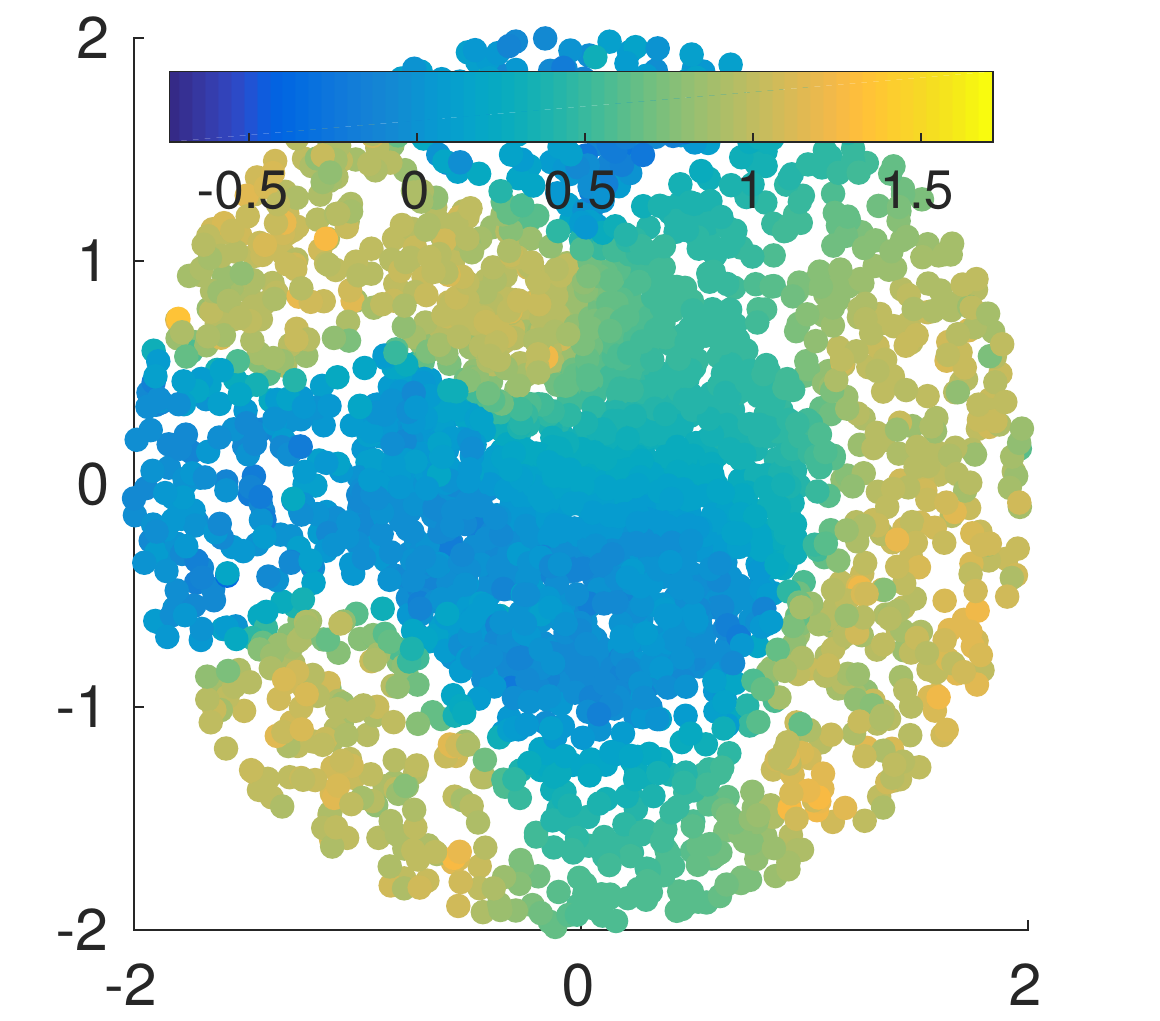}
  	\put(-3,83){\footnotesize c.)}
 \end{overpic} 
      \begin{overpic}[width=0.24\textwidth]{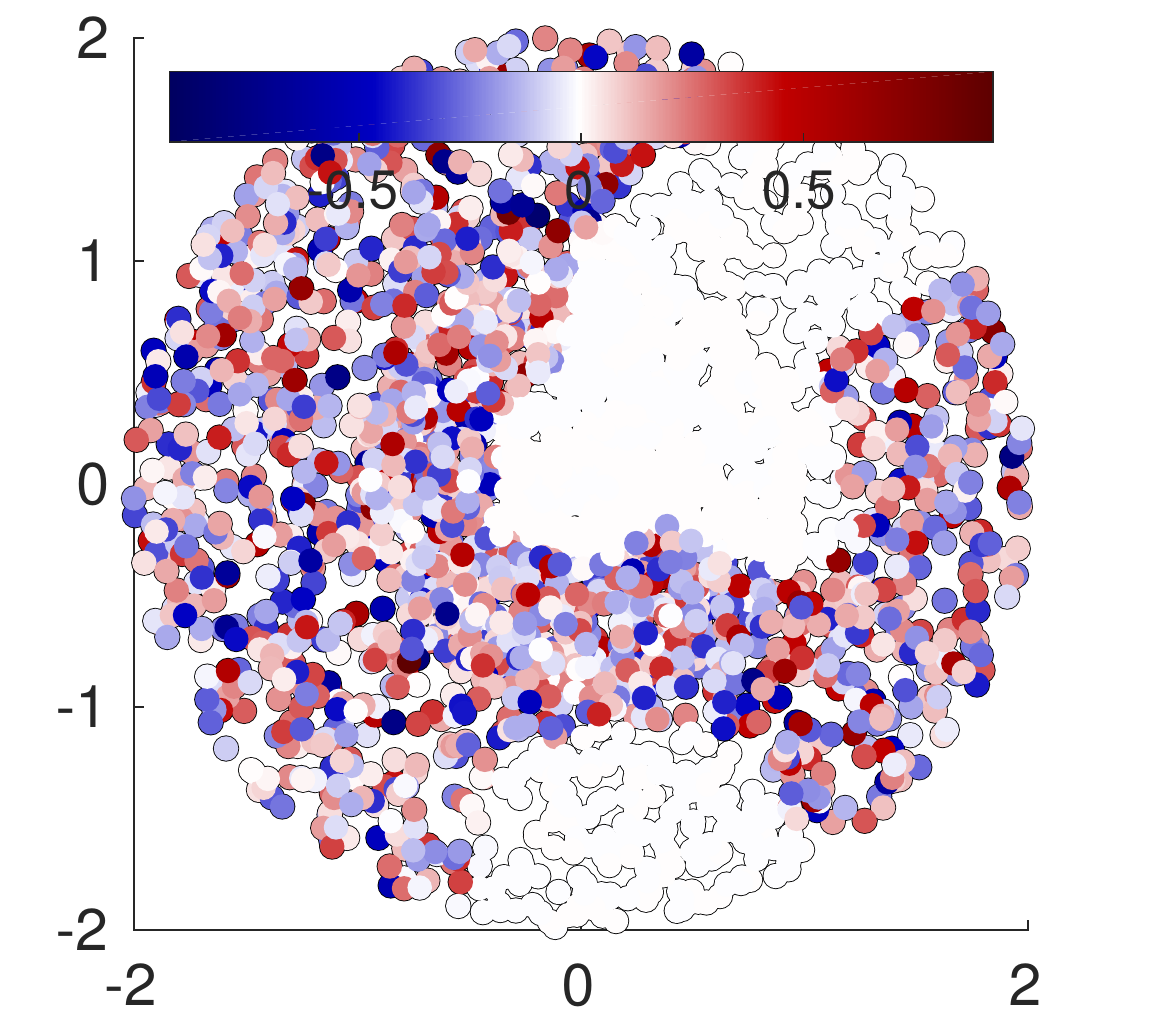}
  	\put(-3,83){\footnotesize d.)}
 \end{overpic} 
  \\  \vspace{2 mm}
      \begin{overpic}[width=0.24\textwidth]{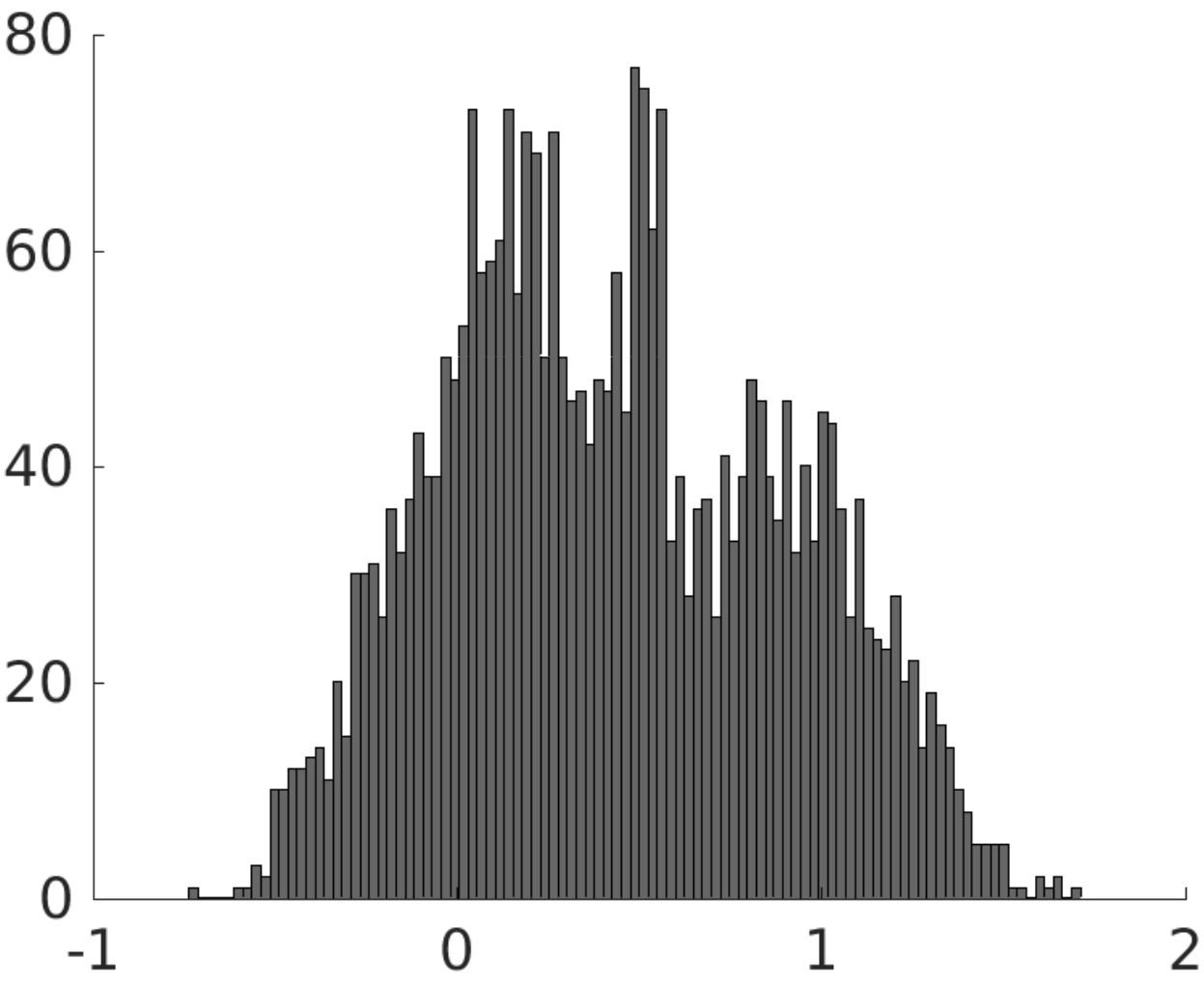}
  	\put(-3,83){\footnotesize e.)}
  \end{overpic} 
    \begin{overpic}[width=0.24\textwidth]{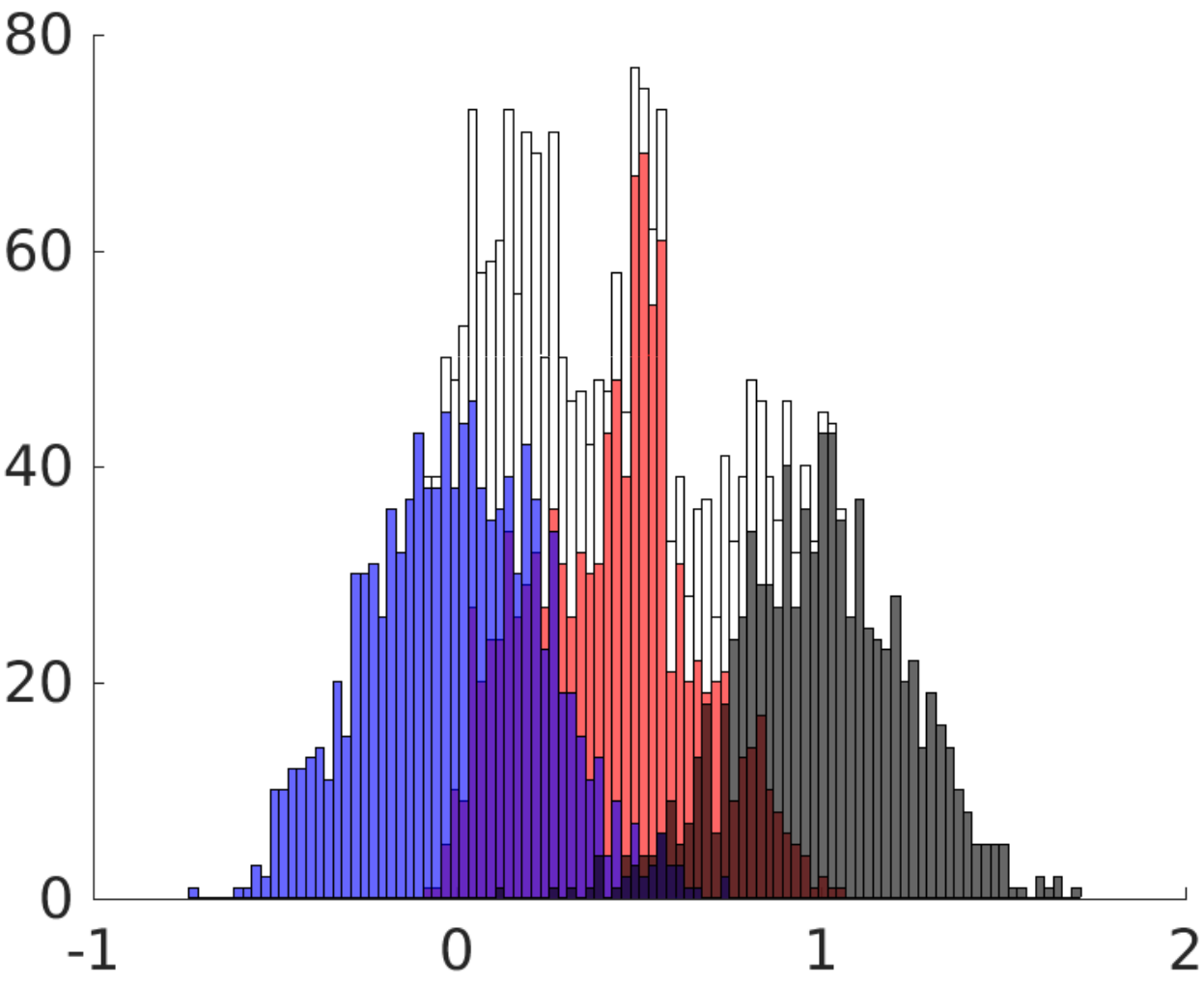}
  	\put(-3,83){\footnotesize f.)}
 \end{overpic} 
    \begin{overpic}[width=0.24\textwidth]{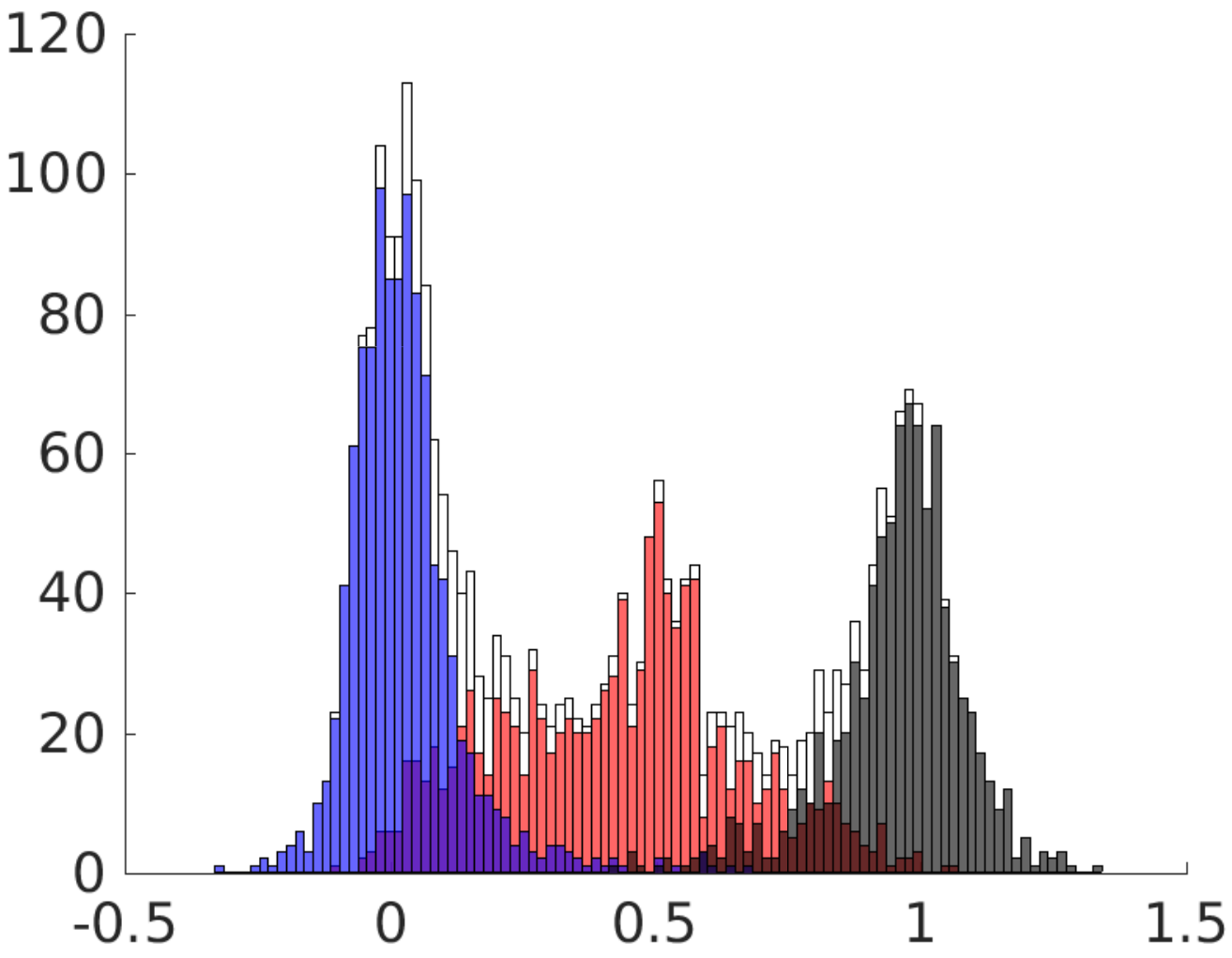}
  	\put(-3,83){\footnotesize g.)}
  \end{overpic} 
      \begin{overpic}[width=0.24\textwidth]{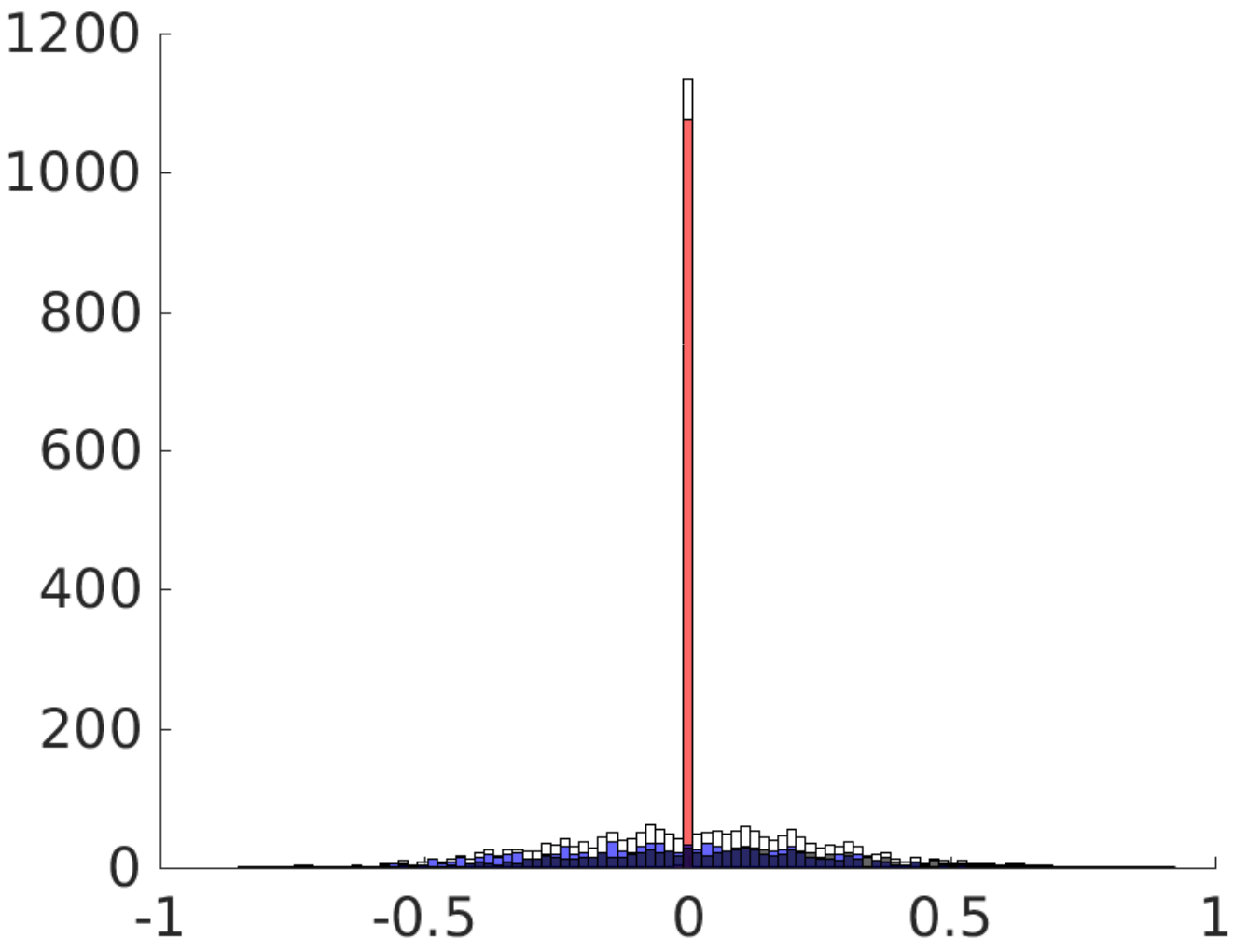}
  	\put(-3,83){\footnotesize h.)}
  \end{overpic} 
  \caption{\footnotesize{Plots relating to toy model. (a,e.) Plot of \emph{observed} cell intensities, and corresponding histogram of intensities. (b,f.) Plot of simulated cell locations, now marked red or blue according to cell phenotype, and histogram now broken up by cell phenotype. (c,g.) Plot and histogram of \emph{baseline} cell intensities. (d,h.) Plot and histogram of \emph{cell impact} intensities. For the simulation shown, $\alpha_* = 10.91$, $R^2_{\max} = 79\%$. Using $k$-means clustering on the 3 data sets ($\cvec$, $\bvec$, $\fvec$), cell phenotypes are correctly identified with: 79\% ($\cvec$), 84\% ($\bvec$), and 13\% ($\fvec$) accuracy.}}
  \label{fig_toy_model}
\end{figure}

\section{Discussion}\label{sec_discussion}

We have presented a method to identify and extract diffusible microenvironment signalling from histology slides. We are also able to quantify what percentage of the data is explainable by diffusion mediated signalling. We applied our method to breast cancer histology slides and found that \HER amplified cells with low \HER stain intensity have higher a higher \emph{cell impact} than the population of all low \HER stain intensity cells. Finally the method was evaluated on synthetic data generated by a simpler model for cell behaviour.

This initial study suggests a methodology that could be applied to more general problems. For example mechanotransduction is another important method of cellular communication \cite{Humphrey_2014}, to which our simulated ablation approach could be applied.

\subsection{Method Applicability}

In our analyses of clinical samples, we used stains for proteins where we are unsure of the exact sources behind variation in stain intensity, be that genetic, microenvironmental, or stochastic. \HER amplification is widely believed to be a genetic event and should be heritable. Expression levels of \HER protein are expected to reflect the gene amplification status. On the other hand, some experimental evidence shows \HER expression might reflect not only genetic material, but also microenvironmental stimuli, such as Notch and NF-$\kappa$B (RANK) signalling (with the effect of increased \HER transcription and regulation of the expansion of cancer stem cells) \cite{Zhang_2013,Korkaya_2009,Magnifico_2009,Ithimakin_2013}. The interpretation of the $R^2_{\max}$ values should be as \emph{what is the maximum percent of the signal variance explainable via our model of cellular communication for paracrine signalling}. Therefore, our method can potentially reveal paracrine interaction even in scenarios where variability of the analysed trait is primarily believed to reflect genetic differences.

Due to this work being only a preliminary study, the presented analysis should be viewed as a promising first step and demonstration of the method, rather than bona fide proof of principle. We envision that the method presented here will be directly applicable to defined scenarios of biological and clinical importance. For example, paracrine signalling is responsible for microenvironment-directed therapy resistance against most of the targeted anti-cancer therapies used in clinics \cite{Wilson_2012}. Yet, the signalling fields as well as impact of individual cells on them have not been studied due to the lack of appropriate tools. For example, the method could be used to interrogate the spatial distribution of c-MET phosphorylation, implied in resistance to ALK targeting tyrosine kinase inhibitors in non-small cell lung cancers \cite{Yamada_2012, Wilson_2012}, as c-MET phosphorylation should be reflective of microenvironmental gradients of its ligand HGF, which is primarily produced by cancer-associated fibroblasts.

A clear next step to promote use and acceptance of our method would be to see how our method performs against different types of cancer stained under the same protocol. Particularly, it may be of particular interest to investigate cancers that are known to be genetically homogeneous, and so the majority of the variance in observation may come from diffusion. In contract, one could investigate genetically heterogeneous cancers and therefore proportionally less variance may be accountable by diffusion. Additionally, it would be particularly pertinent to design experiments where stains relating to metabolic activity were selected. Cell types are usually identifiable by eye, and our approach may even aid hypothesis of cell function.

\subsection{Data Limitations}

Clearly our method as it stands suffers many drawbacks. Regarding use of data, we are stuck working in two dimensions. Working in three dimensions using complete reconstructions of cell geometry would be possible; however, this would be expensive and not repeatable in a clinical setting. 

There will also be edge effect artefacts that we have not accounted for in the model. By this, we mean there will be cells that impact the SF, but were excluded by the biopsy extraction and preparation process. One option is to decrease sample sizes by ignoring a layer of cells at the edge of the histology slice.

When considering the data set used in Sections \ref{sec_molpath_results}, for our method to ``pick out'' important features of the data, we recommend that $R^2_{\max} \gtrsim 40\%$ at a minimum.

\subsection{Technique Refinement}\label{subsec_technique_refinement}

It is likely that other methods for parameter selection are also worthy of exploration. For example, one could also make the opposite assumption: that the contribution of genetic heterogeneity is large, and the microenvironment minimally contributes to the \emph{observed} data set $\cvec$ --- however this does not work practically as then one could then either set $\alpha \to\infty$, or $\gamma = 0$ and then $\bv_i=0$ and $\Vert \cvec - \bvec \Vert$ is maximised. Constraints have to be introduced in an intelligent manner.

Our model could also be expanded to include different classes of objects, for instance, it would not be difficult to include blood vessel structures. Additionally, were it the case that a cell was stained multiple times and one had an \emph{a priori} knowledge for how a cell was supposed to function, relevant constraints could be included in the parameter selection method. In this paper, we did not carry out an exhaustive search on parameter selection techniques.

\section{Acknowledgements}
J.P.T-K. received funding from the EPSRC under grant reference number EP/G037280/1. We thank Jacob Scott, Jan Poleszczuk, Shalla Hanson, Marc Ryser, Eric Lau, Aaron Goldman, Alexander R.A. Anderson, David Robert Grimes, Scott Dawson, Peter Koltai, Stefan Klus and Yoganand Balagurunathan for helpful discussions. 



\bibliographystyle{siam}
\bibliography{References_MathBasisMolPath.bib}

%
\appendix 
%

\section{Existence and Uniqueness}\label{app_exist_unique}

The original problem is stated as equations \eqref{eq_new_chem_prof}--\eqref{eq_new_chem_prof_bc} for each $i\in\mathcal{N}$. Without any loss of generality, we can solve equations \eqref{eq_gen_chem_prof}--\eqref{eq_gen_chem_prof_bc}. Written out again, this is
\begin{align}
\DelX u -\alpha^2 u=0      & \text{ in } \Om     \, ,  \nonumber \\
\vect{n}_{i} \cdot \NabX u = \gamma (c_{i} - u)  & \text{ on } \p \CellDom_{i} \, .
\label{eq_systGen}
\end{align}
We specify that $\p \CellDom_i$ is $\mathcal{C}^1$ regular, and we prove the existence and uniqueness of $u$ in $H^1(\Om)$ --- although the solution is likely smoother.

The question whether we have to impose a limit condition of the form $\lim_{\Vert \xv\Vert  \rightarrow\infty} u (\xv)=0$ is \emph{a priori} unclear. Intuitively, if we impose such a condition, it has to be zero, since the dissipation term would mean a null concentration at an infinite distance of the source. In fact, we will not have impose this, as the unique solution of \eqref{eq_systGen} is zero at infinity, due to the dissipation.

Let us recall that the Sobolev space $W_{m}^p(\Om)$ is the space of function defined as follows
\begin{equation}
W_{m}^p(\Om)=\{ u\in L^p(\Om) \, \vert \,  D^{\beta}u\in L^p(\Om)\, ,\forall \beta \leq m\} \, .
\end{equation}
Here we will be particularly focused on the Hilbert space $H^1(\Om)=W_{1}^2(\Om)$ endowed with the $H^1$ norm defined by
\begin{equation}
\Vert v\Vert_{H^1(\Om)}=\Vert v \Vert_{L^2(\Om)}+\Vert \nabla v \Vert_{L^2(\Om)}\, ,  \,\forall v\in H^1(\Om) \, .
\end{equation}
\begin{prop}
The problem \eqref{eq_systGen} has a unique solution in $H^1(\Om)$.
\end{prop}
\begin{proof}
The weak formulation of the problem \eqref{eq_systGen} given by
\begin{equation}
a(u,v) = l(v)\, , \quad \forall v\in H^{1}(\Om) \, ,
\label{eq_weakSyst}
\end{equation}
for
\begin{equation}
a(u,v)=\int_{\Om} \NabX u(\xv)\cdot\NabX v (\xv) \dd \xv+\alpha^2\int_{\Om}u(\xv)v(\xv) \dd \xv+\gamma\int_{\p \Om}u(\xv)v(\xv)\dd S\, , \quad \forall u,v\in H^1(\Om) \, ,
\end{equation}
and
\begin{equation}
l(v)=\gamma\int_{\p \Om}\psi(\xv) v(\xv)\dd S \, , \quad \forall v\in H^1(\Om)\,  ,
\end{equation}
where $\p \Om=\bigcup_{i}\p\CellDom_{i}$ and $\psi$ is a function satisfying
\begin{equation}
\psi( \xv )=\cv_{i}  \text{ when } \xv \in \CellDom_{i} \, .
 \end{equation}
We now aim to prove continuity and coercivity of $a(u,v)$ on $H^1(\Om)$, as well as the continuity of $l(v)$ on $H^1(\Om)$ to satisfy the Lax--Milgram Theorem (Theorem 5.8 on page 83 in \cite{Gilbarg_2015}). Specifically, we wish to show that
\begin{eqnarray}
\text{Continuity of $a(u,v)$}\, & \iff &\, \exists A>0 \text{ s.t. }\vert a(u,v) \vert \leq  A \Vert u \Vert \Vert v \Vert \, ,  \nonumber \\ 
\text{Coercivity of $a(u,v)$}\, & \iff &\, \exists B>0 \text{ s.t. }\vert a(u,u) \vert \geq  B \Vert u \Vert   \, ,\\ 
\text{Continuity of $l(v)$}\, & \iff &\, \exists C>0 \text{ s.t. }\vert l(v) \vert \leq  C \Vert v \Vert  \, .  \nonumber
\end{eqnarray}

\underline{Continuity of $a(u,v)$}

For $u,v\in H^1(\Om)$, by nature of the modulus function, one can immediately write following inequality
\begin{equation}
\vert a(u,v) \vert \leq \left\vert\int_{\Om} \NabX u(\xv)\cdot\NabX v(\xv)\dd \xv\right\vert+\alpha^2\left\vert\int_{\Om}u(\xv)v(\xv)\dd \xv\right\vert+\gamma\left\vert\int_{\p\Om}u(\xv)v(\xv)\dd S\right\vert.
\end{equation}
The Cauchy-Schwartz inequality allows us then to write
\begin{equation}
\vert a(u,v) \vert \leq \Vert\NabX u\Vert_{L^2(\Om)}\Vert\NabX v\Vert_{L^2(\Om)}+\alpha^2\Vert u\Vert_{L^2(\Om)}\Vert v\Vert_{L^2(\Om)}+\gamma\Vert u_{|\p \Om}\Vert_{L^2(\p\Om)}\Vert v_{|\p \Om}\Vert_{L^2(\p\Om)}.
\label{eq_BilinInequ}
\end{equation}
According to the Trace Theorem \cite{Ding_1996}, since $u,v\in H^1(\Om)$ and because $\Om$ is $\mathcal{C}^1$ and $\p\Om$ is bounded, then $u_{|\p \Om},v_{|\p \Om}\in H^{\frac{1}{2}}(\p\Om)$ and the Trace operator is continuous from $H^1(\Om)$ to $H^{\frac{1}{2}}(\p\Om)$. So, there exists a positive constant $C$ such that
\begin{equation}
\Vert u_{|\p \Om}\Vert_{H^{\frac{1}{2}}(\p\Om)}\Vert v_{|\p \Om}\Vert_{H^{\frac{1}{2}}(\p\Om)}\leq C^2\Vert u\Vert_{H^1(\Om)}\Vert v\Vert_{H^1(\Om)}.
\label{eq_TraceIneq}
\end{equation}
 According to the Sobolev Embedding Theorem \cite{Ziemer_2012}, for an open set $U\in\mathds{R}^N$, $W^{m,p}(U)\subset L^q(U), \forall 1<q<\infty$ if $mp=N$, which is the case here for $U=\p \Om$, $m=\frac{1}{2}$, $p=2$ and $N=1$. Moreover, this injection is continuous and then, for all $1<q<\infty$, there is a positive constant $D_{q}$ such that
\begin{equation}
\Vert u_{|\p \Om}\Vert_{L^{q}(\p\Om)}\Vert v_{|\p \Om}\Vert_{L^{q}(\p\Om)}\leq D_{q}^2\Vert u_{|\p \Om}\Vert_{H^{\frac{1}{2}}(\p\Om)}\Vert v_{|\p \Om}\Vert_{H^{\frac{1}{2}}(\p\Om)}.
\label{eq_SobIneq}
\end{equation}
Therefore, using \eqref{eq_TraceIneq} and \eqref{eq_SobIneq}, we can write
 \begin{equation}
\Vert u_{|\p \Om}\Vert_{L^{q}(\p\Om)}\Vert v_{|\p \Om}\Vert_{L^{q}(\p\Om)}\leq C^{2}D_{2}^{2}\Vert u\Vert_{H^1(\Om)}\Vert v\Vert_{H^1(\Om)}.
\end{equation}
From this inequality, we can write \eqref{eq_BilinInequ} as follows
\begin{equation}
\vert a(u,v) \vert \leq \Vert\NabX u\Vert_{L^2(\Om)}\Vert\NabX v\Vert_{L^2(\Om)}+\alpha^2\Vert u\Vert_{L^2(\Om)}\Vert v\Vert_{L^2(\Om)}+ \gamma C^{2}D_{2}^{2} \Vert u\Vert_{H^1(\Om)}\Vert v\Vert_{H^1(\Om)}\, ,
\label{eq_IneqCont}
\end{equation}
Because $\Vert\NabX u\Vert_{L^2(\Om)}\Vert\NabX v\Vert_{L^2(\Om)}\leq \Vert u\Vert_{H^1(\Om)}\Vert v\Vert_{H^1(\Om)}$ and $\Vert u \Vert_{L^2(\Om)}\Vert\ v \Vert_{L^2(\Om)}\leq \Vert u \Vert_{H^1(\Om)}\Vert v \Vert_{H^1(\Om)}$, then the following inequality can be deduced from \eqref{eq_IneqCont}
\begin{equation}
\vert a(u,v) \vert \leq (1+\alpha^2+ \gamma C^{2}D_{2}^{2} )\Vert u\Vert_{H^1(\Om)}\Vert v\Vert_{H^1(\Om)}\, ,
\end{equation}
allowing to conclude that the bilinear form $a(u,v)$ is continuous on $H^1(\Om)$.

\underline{Coercivity of $a(u,v)$}

Let us prove now the coercivity of $a$. Let $u\in H^1(\Om)$. We can write
\begin{equation}
\vert a(u,u) \vert=\Vert\NabX u\Vert_{L^2(\Om)}^2+\alpha^2\Vert u\Vert_{L^2(\Om)}^2+\gamma\int_{\p \Om}[u(\xv)]^2\dd S \, .
\end{equation}
Since the term $\gamma\int_{\p \Om}u^2\dd S$ is positive, then
\begin{equation}
\vert a(u,u) \vert\geq \Vert\NabX u\Vert_{L^2(\Om)}^2+\alpha^2\Vert u\Vert_{L^2(\Om)}^2 \, ,
\end{equation}
and therefore
\begin{equation}
\vert a(u,u) \vert\geq \min (1,\alpha^2)\Vert u\Vert_{H^1(\Om)}^2\, ,
\end{equation}
which proves the coercivity of $a(u,v)$.

\underline{Continuity of $l(v)$}

Let $v\in H^1(\Om)$, we can immediately write the following inequality
\begin{equation}
\vert l (v)\vert\leq\gamma\int_{\p \Om}\vert\psi(\xv) v(\xv) \vert\dd S \, .
\end{equation}
Because the function $\psi$ is bounded on $\Om$, we can write
\begin{equation}
\vert l (v)\vert\leq \gamma\Vert \psi\Vert_{L^{\infty}}\int_{\p \Om}\vert v(\xv)\vert\dd S \, ,
\end{equation}
which can be written
\begin{equation}
\vert l (v)\vert\leq \Vert \psi\Vert_{L^{\infty}}\Vert v\Vert_{L^1(\p\Om)} \, .
\label{eq_IneqCont2}
\end{equation}
Using again the Trace Theorem \cite{Ding_1996}, since $v \in H^1(\Om)$, we know that there exists a positive constant $C$ such that 
\begin{equation}
\Vert v_{|\p \Om}\Vert_{H^{\frac{1}{2}}(\p\Om)}\leq C\Vert v\Vert_{H^1(\Om)}  \, .
\label{eq_TraceIneq2}
\end{equation}
According once more to the Sobolev Embedding Theorem \cite{Ziemer_2012}, there is a positive constant $D_{q}$ such that
\begin{equation}
\Vert v_{|\p \Om}\Vert_{L^{q}(\p\Om)}\leq D_{q}\Vert v_{|\p \Om}\Vert_{H^{\frac{1}{2}}(\p\Om)} \, .
\label{eq_SobIneq2}
\end{equation}
Therefore, using \eqref{eq_TraceIneq2} and \eqref{eq_SobIneq2}, we can write
 \begin{equation}
\Vert v_{|\p \Om}\Vert_{L^{1}(\p\Om)}\leq C D_{1}\Vert v\Vert_{H^1(\Om)} \, .
\end{equation}
Then we can write \eqref{eq_IneqCont2} as follows
\begin{equation}
\vert l (v)\vert\leq \gamma C D_{1}\Vert \psi\Vert_{L^{\infty}} \Vert v\Vert_{H^1(\Om)} \, ,
\label{eq_IneqCont3}
\end{equation}
which proves the continuity of $l$ on $H^1(\Om)$. Then the Lax--Milgram Theorem \cite{Gilbarg_2015} ensures that the problem \eqref{eq_weakSyst} has a unique solution $u$ on $H^1(\Om)$.
\end{proof}

\begin{remark}
The previous proof does not work for the particular case $\gamma=1$, which corresponds to a Dirichlet condition imposed on the surface of each cell. The proof of existence and uniqueness in this particular case would need to proceed differently. This case can be treated in defining the closed and convex set $K=\{ v\in H^1(\Om)\, \vert \, v-\psi\in H_{0}^1(\Om)\}$ and applying the Stampachia Theorem \cite{Brezis_1983}.
\end{remark}

\section{Numerical Approach}\label{app_numerical_approach}

Due to the geometrical effects involved in this problem, we use a Finite Element Method (FEM) approach. For a practical guide to implementation, see Ref.~\cite{Alberty_1999}, and a theoretical guide to elliptic PDEs, see Ref.~\cite{Gilbarg_2015}. In weak form, our problem is for cell $i$ removed on domain $\Om_i =\Om  \cup \CellDom_i =  \mathds{R}^2 \setminus \bigcup_{j\neq i} \CellDom_j$. For $u_i,v \in H^1(\Om_i )$, we have the weak form problem
\begin{align}
\int_{\Om_i} \NabX v(\xv) \cdot \NabX u_i (\xv) \dd \xv + \alpha^2 \int_{\Om_i}  v(\xv) u_i (\xv) \dd \xv + \gamma  \int_{\p\Om_i}  v(\xv) u_i (\xv) \dd S = \gamma \int_{\p\Om_i}  v(\xv) \psi_i (\xv)\dd S \, ,
\end{align}
where $\psi_i = \psi_i(\xv)$ is any function that satisfies
\begin{equation}
\psi_i(\xv) = \cv_j \text{ when } \xv\in\bar{\CellDom}_j \, ,
\end{equation}
for $j = 1,\dots,i-1,i+1,\dots,N$. We use a standard continuous Galerkin method with piecewise linear basis functions $\{\eta_k(\xv) \}_{k = 1}^K$ on the triangulation of $\Om_i$, $\mathcal{T} = \mathcal{T}(\Om_i)$. For a triangle $T$ with vertices $(t_1,t_2,t_3)$, and the $k^{\text{th}}$ basis function being located at vertex $t_1$, $\eta_k$ is given as 
\begin{equation}
\eta_k (x,y) =  \frac{1}{2 \vert T \vert } \det \left( \begin{array}{ccc} 1 & x & y \\ 1 & x_{t_2} & y_{t_2} \\ 1 & x_{t_3} & y_{t_3}  \end{array} \right) \, , \quad  \forall(x,y)\in T   \, ,
\end{equation}
for 
\begin{equation}
\vert  T \vert = \text{Area}(T)  =  \frac{1}{2} \det \left( \begin{array}{ccc} 1 & x_{t_1} & y_{t_1} \\ 1 & x_{t_2} & y_{t_2} \\ 1 & x_{t_3} & y_{t_3}  \end{array} \right) \, ,
\end{equation}
and $\eta_k (x,y) = 0$ for $(x,y)\not\in T$. Therefore
\begin{equation}
\eta_k (x_l, y_l) = \de_{k l} \, , \text{ for }k,l = 1,\dots K \, .
\end{equation}
Replacing $v$ by the $l^{\text{th}}$ basis function $\eta_l$ and expanding $u_i$ as the sum
\begin{equation}
u_i = \sum_{k} a_k \eta_k \, ,
\end{equation}
we obtain a linear system of equations
\begin{equation}
\left( L + \alpha^2 D + \gamma R \right) \vect{a} =   \gamma \vect{r} \, ,
\end{equation}
and we solve for $\vect{a}$. Notice that changing constants $(\alpha,\gamma)$ do not require reassembly of matrices. The matrix entries are given as follows: the Laplacian matrix $L$ is given by
\begin{align}
L_{kl} & = \int_{\Om_i} \NabX \eta_k (\xv) \cdot \NabX \eta_l (\xv)\, \dd \xv \, ;
\end{align}
the exponential decay matrix $D$ is given by
\begin{align}
D_{kl} & = \int_{\Om_i} \eta_k (\xv) \,\eta_l (\xv) \,\dd \xv \, ;
\end{align}
and the robin boundary condition comes in two parts, first the matrix $R$,
\begin{align}
R_{kl} & =    \int_{ \p \Om_i } \eta_k (\xv) \, \eta_l (\xv)\, \dd S \, , 
\end{align}
and then the vector $\vect{r}$ with entries
\begin{align}
r_k  & = \int_{\p \Om_i} \eta_k (\xv) \, \psi_i (\xv)\,\dd S \, .
\end{align}
Within this finite element framework, calculating the \emph{baseline} is then 
\begin{align}
\bv_i &= \sum_{T\in \mathcal{T}(\CellDom_i)} \int_T a_k \eta_k(\xv) \dd \xv \, ,\\
     &= \sum_{T\in \mathcal{T}(\CellDom_i)}  \frac{\vert T\vert }{3} [a_{t_1} + a_{t_2} + a_{t_3}] \, .
\end{align}

Calculating these inner products between basis functions can be challenging. We now give a practical guide to implementation.

\subsection{Practical Approach to FEM Implementation}

To practical implement our FEM scheme, instead of putting in the entries to the matrices $L, B, R$ and vector $b$ in individually, it is much simpler to do it each triangle or boundary edge at a time. We write $L$ and $B$ as a sum over the triangles $T$ in the triangulation $\mathcal{T}(\Om_i)$
\begin{align}
L &= \sum_{T \in \mathcal{T}(\Om_i)} L^{(T)} \, ,  \\ 
D &= \sum_{T \in \mathcal{T}(\Om_i)} D^{(T)} \, .
\end{align}
The entries of $L^{(T)}$ are given as
\begin{align}
L_{kl}^{(T)} = \left\{ \begin{array}{cc} \tilde{L}_{t_k, t_l}^{(T)} & \text{if vertices $(k,l)$ are part of triangle $T$}  \\ 0 & \text{otherwise} \end{array} \right\} \, ,
\end{align}
and analogously for $D^{(T)}$
\begin{align}
D_{kl}^{(T)} = \left\{ \begin{array}{cc} \tilde{D}_{t_k, t_l}^{(T)} & \text{if vertices $(k,l)$ are part of triangle $T$}  \\ 0 & \text{otherwise} \end{array} \right\} \, ,
\end{align}
and therefore $\tilde{L}$ and $\tilde{D}$ are $3\times 3$ matrices. These matrices have well known analytic solutions given as
\begin{align}
\tilde{L} = \frac{| T |}{2} G\, G^\dagger \, ,
\end{align}
for
\begin{align}
G = \left( \begin{array}{ccc} 1 & 1 & 1 \\ x_{t_1} & x_{t_2} & x_{t_3} \\ y_{t_1} & y_{t_2} & y_{t_3} \end{array} \right)^{-1} \left( \begin{array}{cc} 0 & 0 \\ 1 & 0 \\ 0 & 1 \end{array} \right) \, ,
\end{align}
and 
\begin{align}
\tilde{D} =   \frac{| T |}{12}  \left( \begin{array}{ccc} 2 & 1 & 1 \\ 1 & 2 & 1 \\ 1 & 1 & 2  \end{array} \right) \, .
\end{align}

For matrix $R$ and vector $\vect{r}$, we sum over edges $E$ that form the discretised boundary $\mathcal{E}(\p\Om_i)$
\begin{align}
R &= \sum_{E \in \mathcal{E}(\p\Om_i)} R^{(E)} \, , \\ 
\vect{b} &= \sum_{E \in \mathcal{E}(\p\Om_i)} \vect{b}^{(E)} \, ,
\end{align}
where 
\begin{align}
R_{kl}^{(E)} = \left\{ \begin{array}{cc} \tilde{R}_{e_k, e_l}^{(E)} & \text{if vertices $(k,l)$ are part of edge $E$}  \\ 0 & \text{otherwise} \end{array} \right\} \, ,
\end{align}
and therefore $\tilde{R}$ is a $2\times 2$ matrix given as
\begin{align}
\tilde{R} = \frac{| E |}{6} \left( \begin{array}{cc} 2 & 1 \\ 1 & 2 \end{array} \right)\, ,
\end{align}
for $\vert E \vert = \text{Length}(E) = \Vert \xv_{e_1} - \xv_{e_2} \Vert$. Similarly for entries of vector $\vect{r}$
\begin{align}
r_k = \left\{ \begin{array}{cc} \tilde{r}_{e_k}^{(E)} & \text{if vertex $k$ are part of edge $E$}  \\ 0 & \text{otherwise} \end{array} \right\} \, ,
\end{align}
and therefore 
\begin{align}
\tilde{r}_k = \frac{| E |}{2} \psi_i (\xv_{e}) \, .
\end{align}

\end{document}